\documentclass[journal,10pt]{IEEEtran}
\usepackage{ifpdf}
\usepackage{cite}
\usepackage{algorithmic}
\usepackage{algorithm}
\ifCLASSINFOpdf
\usepackage[pdftex]{graphicx}

\else
 \usepackage[dvips]{graphicx}
\fi
\usepackage[cmex10]{amsmath}
\usepackage{epstopdf}
\usepackage{array}
\usepackage{flushend}
\usepackage{balance}
\usepackage{eqparbox}

\usepackage[tight,footnotesize]{subfigure}
\usepackage{fixltx2e}
\usepackage{color}
\usepackage{float}

\usepackage{dblfloatfix}
\usepackage{url}
\usepackage{cite}
\usepackage{subfigure}
\usepackage{amsmath}

\usepackage{epstopdf}

\newenvironment{proof}{{\indent \indent \it Proof:}}{\hfill $\blacksquare$}

\newtheorem{lemma}{Lemma}
\usepackage{amssymb}
\usepackage{color}

\begin{document}
\allowdisplaybreaks[3]
\title{Joint User Association and Beamforming in Integrated Satellite-HAPS-Ground Networks}

\setlength{\columnsep}{0.21 in}

\author{Shasha~Liu,~\IEEEmembership{Student Member,~IEEE,}
        Hayssam~Dahrouj,~\IEEEmembership{Senior Member,~IEEE,}
        and
        Mohamed-Slim~Alouini,~\IEEEmembership{Fellow,~IEEE}

\thanks {Shasha Liu is with the University of Electronic Science and Technology of China (UESTC) (e-mail: shashaliu@std.uestc.edu.cn).

Hayssam Dahrouj and Mohamed-Slim Alouini are with the Division of Computer, Electrical and Mathematical Sciences and Engineering, King Abdullah University of Science and Technology, Thuwal 23955-6900, Saudi Arabia (e-mail: hayssam.dahrouj@gmail.com, slim.alouini@kaust.edu.sa).

{An extended version of the current paper is available on archive \cite{ShashaTWC2022}.}}

}

\maketitle
\begin{abstract}
This paper proposes, and evaluates the benefit of, one particular hybrid satellite-HAPS-ground network, where one high-altitude-platform-station (HAPS) connected to one geo-satellite assists the ground base-stations (BSs) at serving ground-level users. The paper assumes that the geo-satellite is connected to the HAPS using free-space-optical backhaul link. The HAPS, equipped with multiple antennas, aims at transmitting the geo-satellite data to the users via radio-frequency (RF) links using spatial-multiplexing. Each ground BS, on the other hand, is equipped with multiple antennas, but directly serves the users through the RF links. The paper then focuses on maximizing the network-wide throughput, subject to
{HAPS payload connectivity constraint,}
HAPS and BSs power constraints, and backhaul constraints, so as to jointly determine the user-association strategy of each user (i.e., user to geo-satellite via HAPS, or user to BS), and their associated beamforming vectors. We tackle such a mixed discrete-continuous optimization problem using an iterative approach, where the user-association is determined using a combination of integer linear programming and generalized assignment problems, and where the beamforming strategy is found using a weighted-minimum-mean-squared-error approach. The simulations illustrate the appreciable gain of our proposed algorithm, and highlight the prospects of augmenting the ground networks with beamforming-empowered HAPS for connecting the unconnected, and super-connecting the connected.
\end{abstract}
\IEEEpeerreviewmaketitle
\section{Introduction}
\subsection{Overview}
Large-scale ground-level connectivity has a major impact on
current telecommunications infrastructures, which today support billions of people, and tens of billions of devices \cite{saeed2021P2P}.
As this demand is expected to grow at an even faster pace over the next few years,  a major practice of telecom operators is to densify the terrestrial network  infrastructures \cite{alam2021high}. Such densification may, however, not always be able to satisfy the data ultra-hungry devices and their ambitious quality-of-services requirements in high-interference regimes, and is also not feasible to be realized in rural and remote areas.
Augmenting ground-level communications with spatial networks, e.g., satellites at the GEO-layer (i.e., geo-satellites),
and the stratospheric layer (i.e., High-Altitude Platform Systems (HAPS)),
is expected to revolutionize the physical layer paradigm of the sixth generation of wireless systems and beyond (6G and beyond).
By adding HAPS to the traditional space-ground network, some of the shortcomings and challenges of the existing networks can be well-resolved, especially those related to 6G networks challenges and goals towards connecting the unconnected, and over-connecting the connected \cite{saeed2021P2P,alam2021high}. HAPS enhances terrestrial communications due to its better path-loss profile as compared to higher layers platforms, e.g., cubesats \cite{qiu2019air}.
Compared with Unmanned Aerial Vehicles (UAVs), HAPS can cover a more expansive area due to its elevated altitude and wider beam coverage \cite{arum2020review}.
In addition, HAPS is located at the stratospheric layer,
which provides several appealing deployment characteristics, e.g., the ability to maintain a quasi-stationary state and achieve global connectivity \cite{kurt2021vision,mohammed2011role,karapantazis2005broadband}.
The true assessment of such deployment remains, however, a strong function of the joint resource allocation across the HAPS and ground base-stations, and so this paper proposes one particular framework for optimizing integrated satellite-HAPS-ground networks under specific physical connectivity constraints.

This paper considers a vertical heterogeneous network (VHetNet) comprising one geo-satellite, one HAPS, and several terrestrial base-stations, where both the HAPS and ground base-stations (BSs) are equipped with multiple antennas to simultaneously serve multiple users. The paper assumes that the satellite and HAPS are connected through a free-space optical (FSO) link,
which has a wide bandwith, and is secure, license-free, and suitable for deploying point-to-point (P2P)  communication in the space \cite{saeed2021P2P} \cite{fidler2010optical}.
{
Moreover, since HAPS is located in the stratosphere layer, the FSO link connecting the geo-satellite to HAPS does not go through the troposphere layer, which makes the geo-satellite-to-HAPS FSO link less vulnerable to weather conditions, e.g., turbulence, rain, fog, etc. \cite{saeed2021P2P}.}
FSO links, however, requires strong alignment between the sender and receiver \cite{henniger2010introduction}, which is not suitable for the mobile nature of ground users communications. To this end, the paper assumes that both the HAPS and ground BSs communicate with their respective ground users using radio-frequency (RF) links, where each user can be served by either the ground BSs or by the HAPS. {
Unlike ground base-stations, however, HAPS payload consists of three subsystems: flight control system, energy management subsystem, and communication payload system \cite{kurt2021vision}, which poses an additional constraint on the HAPS connectivity capability. More specifically, HAPS requires a flight control system to handle mobility and maintain a quasi-stationary state, an energy management subsystem for energy storage and distribution, and a communication payload system to handle the communication between the HAPS and other entities. Therefore, our paper accounts for such HAPS particularity, by adding the HAPS payload connectivity constraint to the optimization problem at hand.}
 Further, the transmissions across the FSO satellite-HAPS link are assumed to occur over different optical bands, and do not interfere with each other. The considered network performance becomes, therefore, a strong function of the intra-HAPS interference, intra-BS interference, inter-BS interference, and HAPS-BS interference (hereafter denoted by inter-layer interference). The paper then attempts at managing such multi-mode multi-layered interference by means of associating users with BSs or HAPS, and determining their corresponding beamforming vectors so as to maximize the considered integrated satellite-HAPS-ground network throughput.
\subsection{Related Work}
The problem considered in this paper is related to the optimization of vertical heterogeneous networks, and particularly to the resource allocation problem in integrated satellite-HAPS-ground networks. The tackled problem is also related to user scheduling and beamforming problems, which are studied extensively in the past literature, both individually and jointly, especially in the context of classical interference networks optimization.

Optimizing system throughput in interference networks is often a non-convex optimization problem, and so managing the wireless networks radio resources remains a challenging problem in general \cite{liu2012achieving}. Many recent techniques, therefore, aim at devising numerically reasonable optimization algorithms that promise to offer major performance improvements as compared to conventional systems strategies. For example, the user association scheduling problem is considered in several classical networks, e.g., \cite{kim2009interference,wang2008user,reifert2021distributed, shen2018fractional}, all of which focus on terrestrial networks optimization only. Similarly, the joint user association and power assignment problem is addressed in \cite{douik2020mode}; please also see the references therein. The user-association subproblem considered in the current paper, however, involves a more intricate coupling of cross-mode cross-layered interference and HAPS connectivity constraints, and so the paper leverages techniques such as linear integer programming \cite{douik2020tutorial} and generalized assignment problems \cite{Murat} to develop reasonable heuristics for dealing with the problem discrete intricacies.

The problem of beamforming optimization is also extensively studied in the literature of wireless networks, either using Lagrangian-duality \cite{dahrouj2010coordinated}, semidefinite programming (SDP) \cite{dahrouj2011multicell}, weighted-minimum mean squared error (WMMSE) \cite{shi2011iteratively,dai2014sparse1}, or fractional programming (FP) in \cite{shen2018fractional1}. The joint user association and beamforming problem is also addressed in \cite{yu2013multicell,khan2018optimizing} under specific terrestrial systems scenarios. From a methodology perspective, WMMSE and FP are noticeably popular solutions to maximize the sum-rate in conventional terrestrial networks \cite{shi2011iteratively,dai2014sparse1,shen2018fractional1}. Given the structure of our problem formulation, a part of the current paper proposes a tweaked version of WMMSE to best account for the physical constraints stemming from the FSO backhaul link constraint and the multi-mode multi-layered interference in the context of maximizing the system sum-rate under fixed user-association strategy.


The problem considered in this paper is strongly coupled with the latest advances and studies of HAPS networks, which come at the forefront of sky connectivity latest trends. For example, reference \cite{kurt2021vision} presents a comprehensive overview about the vision and framework of HAPS networks. Reference \cite{kurt2021vision} further highlights the prospects of HAPS systems in radio resource management, which our current paper studies under one particular system architecture. In fact, the HAPS used in the current paper also acts as a super-macro base-station, which is well motivated through reference \cite{alam2021high} that illustrates the role of HAPS in serving both remote and metropolitan dense areas.
{The current paper builds upon such connectivity framework, and accounts for specific HAPS payload connectivity constraint, HAPS and ground base-stations power constraints, and backhaul capacity constraints.}


Together with the advances in HAPS studies comes the generic trend in investigating the system-level benefits of integrated space-air-ground networks, often denoted by VHetNets, which aim at achieving a relative digital inclusion through connecting the unconnected and super-connecting the connected \cite{saeed2021P2P}. To this end, references \cite{alzenad2019coverage,cherif2020downlink} focus on the performance analysis of VHetNets and highlight their coverage expansion capabilities. Resource optimization of VHetNets also emerges nowadays as a powerful paradigm for assessing the true benefits of VHetNets in solving the digital divide problem, e.g., \cite{jia2020sum,wang2022resource,alsharoa2020improvement,yahia2021haps}. While reference \cite{jia2020sum} considers the time allocation, power control, and trajectory optimization of UAV-aided space-air-ground networks, reference \cite{wang2022resource} considers power assignment and transmission protocol in an integrated HAPS-mobile telecommunications (IMT) system. References \cite{alsharoa2020improvement,yahia2021haps} are particularly related to our system model; however, our paper goes beyond both \cite{alsharoa2020improvement,yahia2021haps} by accounting for the joint user association and beamforming problem in a cross-mode cross-layered interference setup. More specifically, on the one hand, reference \cite{yahia2021haps} studies a hybrid RF/FSO VHetNet consisting of satellites, HAPS, and ground base-stations, and focuses on the systematic performance analysis of the networks. On the other hand, reference \cite{alsharoa2020improvement} proposes an integrated satellite-airborne-ground network and optimizes the user access, power assignment, and HAPSs' location under an orthogonal frequency division multiple access scheme. Differently from both \cite{alsharoa2020improvement} and \cite{yahia2021haps}, our paper adopts a multiple-antenna scheme at the HAPS and at the ground base-stations, and optimizes the user association and spatial multiplexing strategies so as to efficiently serve the ground users subject to practical system-level constraints.
\subsection{Contributions}
Unlike the aforementioned papers, this paper proposes, and evaluates the benefit of, one particular integrated satellite-HAPS-ground network comprising one satellite, one HAPS, and terrestrial base-stations (BSs), where the geo-satellite is connected to the HAPS using an FSO link. The HAPS and BSs in turn serve ground users using RF links, and are equipped with multiple antennas. While each user is equipped with a single antenna and can be served either by the HAPS or by one of the ground BSs, the user connectivity to the HAPS depends both on the FSO link quality between the HAPS and the satellite, and the HAPS payload capabilities. The paper then addresses the problem of maximizing the network sum-rate, so as to jointly determine the user association strategy, and the beamforming vectors at the HAPSs and BSs,
subject to {HAPS payload connectivity constraint,}
FSO backhaul constraints, and HAPS and BSs maximum power constraints. The paper solves such a mixed-integer non-convex optimization problem in an iterative modular fashion, i.e., it iterates between solving the user association strategy for fixed beamforming, and solving the beamforming problem for fixed user association. The paper contributions can then be summarized as follows:
\begin{itemize}
\item The paper proposes one particular satellite-HAPS-ground multi-antenna network architecture, specifically designed to augment the ground communications through user scheduling and spatial multiplexing. The paper formulates a mixed discrete-continuous optimization problem to jointly pair users with BSs and HAPS, and to design the beamforming vectors at the HAPSs and BSs to maximize the network sum-rate subject to {HAPS payload connectivity constraint,}
    FSO backhaul constraints, and HAPS and BSs maximum power constraints.
\item The paper solves the non-convex optimization problem through an iterative algorithm. That is, we iteratively optimize each of the optimization parameters by fixing other variables. The user association strategy is first determined by linearizing the original problem so as to enable the utilization of
    integer linear programming (ILP), followed by a generalized assignment problem (GAP)-type solution. The beamforming vectors at the HAPS and BSs are then found through a series of problem reformulations that enable the use of weighted minimum mean square error (WMMSE)-type solutions.
\item  The paper compares the proposed algorithm to classical techniques using Monte-Carlo simulations. The paper results illustrate the appreciable sum-rate gain of the proposed joint user association and beamforming algorithm as compared to classical techniques for various network parameters. The simulations particularly highlight the numerical potential of the proposed integrated satellite-HAPS-ground networks optimization framework for connecting the unconnected, and super-connecting the connected, especially at the high interference regime, and under beefed-up HAPS capabilities (i.e., power, number of antennas, quality of FSO backhauling, etc.).
\item The paper draws a handful of design guidelines and recommendations for deploying HAPS in both remote and metropolitan dense areas.
\end{itemize}

The rest of the paper is organized as follows. Section II presents the system model and problem constraints. The problem formulation and proposed algorithm are discussed in Section III. Section IV presents the simulation results that highlight the numerical prospects of the proposed solution. Finally, the paper is concluded in section V.
\section{System Model and Problem Constraints}
\subsection{System Model}
Consider an integrated satellite-HAPS-ground network consisting of one geo-satellite, one HAPS, $N_{B}$ ground BSs, and $N_{U}$ users.
{The paper assumes that the set of transmitters is denoted by $\mathcal{I}=\{0,1,2,...,N_B\}$, where the $0^{th}$ transmitter points to the HAPS.}
We also denote the set of users $\mathcal{U}$ by $\mathcal{U}=\{1,2,...,N_U\}$.
Let $N_{A}^{i}$ be the number of antennas at BSs and HAPS (i.e., $i=0$ for HAPS, $i=1,2 \cdot\cdot\cdot N_{B}$ for terrestrial BSs).
The paper assumes that the satellite communicates with HAPS via the FSO link, while HAPS and BSs connect to the users via RF links. Each user can be served either by the HAPS or by one of the ground base-stations. Being served by HAPS means that the required data is sent from the satellite to HAPS via the FSO link, and then the HAPS sends it to the ground user via the RF link. The ground BS, however, communicates with the user directly. The paper adopts a space division multiplexing scheme, where all RF links use the same central frequency, and where the HAPS and BSs adopt beamforming to serve multiple users simultaneously.
Fig. \ref{SM} shows an example of the considered network, which consists of one geo-satellite, one HAPS, 3 ground BSs and 9 users. Fig. \ref{SM} also illustrates the information flow from the terrestrial gateway to the geo-satellite. The paper, in fact, assumes that such data feeding happens over different time-scales than the considered downlink transmission, i.e., it does not interfere with the considered satellite-HAPS-ground network.
 We next present the channel model and rate expressions of the hybrid space-air-ground system under study.
\begin{figure}[h]
\centering
\includegraphics[width=3in]{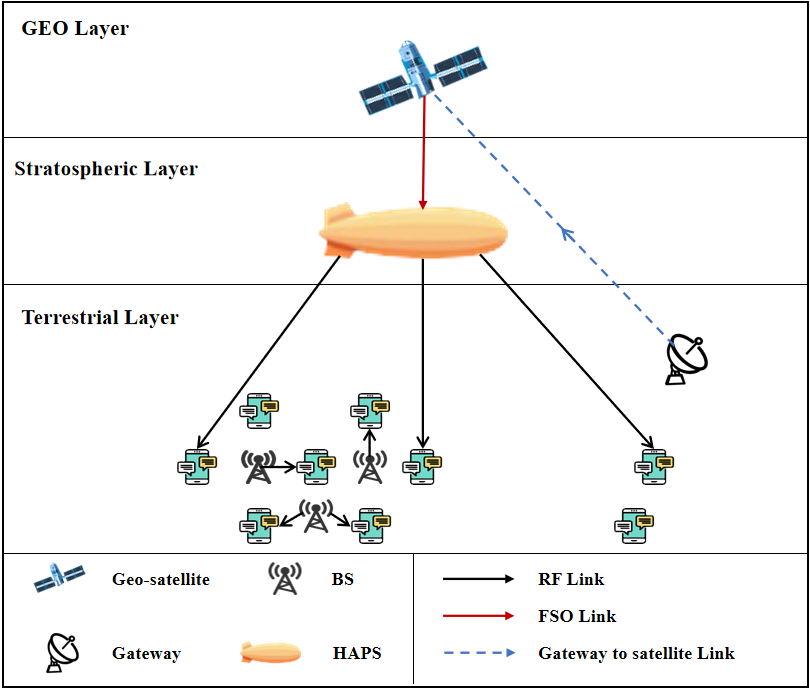}
\caption{An integrated satellite-HAPS-ground network }
\label{SM}
\end{figure}

\begin{table*}[!h]
\centering
\caption{Notations}
\label{tableR}
\begin{tabular}{|p{.1\textwidth} | p{.7\textwidth} | }
\hline
  \textbf{Symbol} &   \textbf{Definition} \\
 \hline
 $N_{B}$ &   The number of terrestrial base-stations\\
  \hline
$N_{U}$ &  The number of users \\
 \hline
$N_{A}^{i}$ & The number of {transmitters antennas} ($i=1,2\cdot\cdot\cdot N_{B}$) or HAPS ($i=0$) \\
\hline
$\gamma_{ij}$ &  Data availability binary variable \\
\hline
$\alpha_{ij}$ &  User association binary variable \\
\hline
$\mathbf{w}_{ij}$ & beamforming vector from {transmitter} $i$ to user $j$ \\
\hline
$R^{FSO}$ &  The achievable data rate at HAPS through the FSO link  \\
\hline
$R_{ij}^{RF}$ &  The achievable data rate of user $j$ served by {transmitter} $i$ through the RF link  \\
\hline
$R_{ij}^{Ground\_BS}$ &  The achievable data rate of user $j$ served by ground BS $i$ through the RF link  \\
\hline
$R_{0j}^{HAPS}$ &  The achievable data rate of user $j$ served by HAPS  \\
\hline
$\hat{R}_{ij}^{RF}$ &  The linearized achievable data rate of user $j$ served by {transmitter} $i$ through the RF link  \\
\hline
$\check{R}_{ij}^{RF}$ &  The interference-free achievable data rate of user $j$ served by {transmitter} $i$ through the RF link  \\
\hline
$\tau_{0j}$ &  Auxiliary variable of achievable data rate of the user $j$ served by HAPS  \\
\hline
$R_{ij}$ &  The achievable data rate of the user $j$ served by {transmitter} $i$  \\
\hline
$\lambda_{0j}$ &  The weight of the $\tau_{0j}$  \\
\hline
$\lambda_{ij}$ &  The weight of the $R_{ij}$  \\
\hline
$\boldsymbol{\rho}_{ij}$ &  The MSE weight for user $j$ served by {transmitter} $i$  \\
\hline
$\mathbf{u}_{ij}$ &  The receive beamforming vector of user $j$ served by {transmitter} $i$ \\
\hline
$\mathbf{e}_{ij}$ &  The MSE for user $j$ served by {transmitter} $i$\\
\hline
\end{tabular}
 \end{table*}
\subsection{FSO Backhaul Capacity}
In this paper, the geo-satellite and the HAPS are connected by the FSO link, the capacity of which strongly depends on the atmospheric attenuation (absorption, scattering). The data rate between the HAPS and the satellite, denoted by $R^{FSO}$, can then be written as \cite{alzenad2018fso}:
\begin{equation}
\label{RSO}
R^{FSO}=\frac{P_{t}\eta_{t}\eta_{r}10^{\frac{-L_{poi}}{10}}10^{\frac{-L_{atm}}{10}}A_{R}}{A_{B}E_{p}\eta_{b}},
\end{equation}
where $P_{t}$ denotes the transmit power of satellite, $\eta_{t}$ and $\eta_{r}$ stand for the optical efficiencies of the transmitter and receiver, respectively, $L_{poi}$ is the pointing loss, $L_{atm}$ is the atmospheric attenuation, $A_{R}$ and $A_{B}$ is the area of the FSO receiver and beam ({i.e., $\frac{A_{R}}{A_{B}}$ is the geometrical loss, which characters the pathloss.}),
$E_{p}$ denotes the photon energy respectively, and $\eta_{b}$ represents the receiver sensitivity.
\subsection{RF Channel Model}
According to \cite{alsharoa2020improvement}, the channel {coefficient} between the $n^{th}$ antenna of the $i^{th}$ {transmitter} ($i=0$ for HAPS, $i=1,2 \cdot\cdot\cdot N_{B}$ for BSs) and the $j^{th}$ user, denoted by $h_{ij,n}$, is given by
\begin{equation}
\label{RF_channel}
h_{ij,n}=\left(\frac{c}{4\pi d_{ij,n}f_{c}}\right){A_{ij,n}}F_{ij,n},
\end{equation}
where c is the speed of light, $f_{c}$ is the carrier frequency, $d_{ij,n}$ is the distance between the $n^{th}$ antenna of the $i^{th}$ {transmitter} and the $j^{th}$ user. For terrestrial BSs, $A_{ij,n}$ corresponds to a log-normal shadowing and $F_{ij,n}$ is the Rayleigh small scale gain. In the case of HAPS, we omit $A_{0j,n}$  and represent $F_{0j,n}$ as Rician small-scale gain denoted by $\kappa_{HAPS}$, due to the strong line-of-sight between the HAPS and the ground users. This, however, is adopted without loss of generality as the optimization framework rather depends on the values of channel vectors. More specifically, in the rest of the paper, we simply denote the RF channel vector between {transmitter} $i$ and user $j$ ($i=0$ for HAPS, $i=1,2 \cdot\cdot\cdot N_{B}$ for BSs) as $\mathbf{h}_{ij}\in \mathbb{C}^{N_{A}^{i}}$, where  $\mathbf{h}_{ij}=[h_{ij,1}, h_{ij,2}, \cdot\cdot\cdot , h_{ij,n}, \cdot\cdot\cdot ,h_{ij,N_{A}^{i}}]^{T}$.

\subsection{User Association Scheme}
The paper considers the practical consideration that user $j$ request may (or may not) be available at transmitter $i$. To this end, we introduce the binary variable $\gamma_{ij}$ which is defined as 1 if data required by user $j$ is available at {transmitter} $i$, and zero otherwise, $\forall i\in \mathcal{I}$ and $j\in \mathcal{U}$. We note that the variables $\gamma_{ij}$ are fixed in the context of our paper, and are known to the optimizer. The paper then assumes that each ground user can be served by one transmitter at most (i.e., either by the HAPS or by one of the ground BSs.).
{Furthermore, to account for the HAPS payload connectivity constraint, we denote by $K_{0}$ the maximum number of users that the HAPS can serve.}
Introduce a binary variable $\alpha_{ij}$, which is equal to $1$ if the user $j$ is served by the {transmitter} $i$ and $0$ otherwise, which yields the following connectivity constraints:
\begin{equation}
\label{CBS1}
   \sum_{j=1}^{N_{U}}\gamma_{0j}\alpha_{0j} \leq K_{0},
\end{equation}
\begin{equation}
\label{CBS2}
  \sum_{i=0}^{N_{B}}\gamma_{ij}\alpha_{ij} \leq 1,\ \forall j \in \mathcal{U}.
\end{equation}
\subsection{Rates Expressions }
This paper considers that multiuser downlink transmit beamforming is employed at both the HAPS and the ground BSs. Let $s_{ij}$ represent the information signal for user $j$ when served by {transmitter} $i$, $\forall i\in \mathcal{I}$ and $j\in \mathcal{U}$, and let $\mathbf{w}_{ij}\in \mathbb{C}^{N_{A}^{i}}$ be the beamforming vector associated with $s_{ij}$. Therefore, the received signal at the $j^{th}$ user, denoted as $y_{j}$, is given by
\begin{equation}
\label{sBS}
y_{j}=\sum_{b=0}^{N_{B}}\sum_{u=1}^{N_{U}}\gamma_{bu}\alpha_{bu}s_{bu}\mathbf{h}_{bj}^{H}\mathbf{w}_{bu} + z_{j},
\end{equation}
where $\mathbf{h}_{bj}=[h_{bj,1}, h_{bj,2}, \cdot\cdot\cdot , h_{bj,n}, \cdot\cdot\cdot , h_{bj,N_{A}^{b}}]^{T}$  is the vector channel from {transmitter} $b$ to user $j$, and where $z_{j}$ is the additive white circularly symmetric Gaussian complex noise with variance $\frac{\sigma^{2}}{2}$ on each of its real and imaginary components.

The above expression (\ref{sBS}) has implicity four types of interference, namely, intra-HAPS interference, intra-base-station interference, inter-layer interference, and inter-base-station interference. 

We next present the rates expressions of each user $j$ according to the two types of user-association possibilities. If user $j$ is served by {transmitter} $i$, the associated signal-to-interference-plus-noise ratio (SINR), denoted by ($\mathrm{SINR}_{ij}$), can be expressed as:
\begin{equation}
\label{SINR}
\mathrm{SINR}_{ij}=\frac{|\mathbf{h}_{ij}^{H}\mathbf{w}_{ij}|^{2}}{\sum_{u=1,u\neq j}^{N_{U}}\sum_{b=0}^{N_{B}}\gamma_{bu}\alpha_{bu}|\mathbf{h}_{bj}^{H}\mathbf{w}_{bu}|^{2} + \sigma^{2}}.
\end{equation}
The user achievable RF rate can then be written as:
\begin{equation}
\label{R_RF}
  R_{ij}^{RF}\!=\!\beta\log_{2}\left(1\!+\!\frac{|\mathbf{h}_{ij}^{H}\mathbf{w}_{ij}|^{2}}{\sum_{u=1,u\neq j}^{N_{U}}\sum_{b=0}^{N_{B}}\gamma_{bu}\alpha_{bu}|\mathbf{h}_{bj}^{H}\mathbf{w}_{bu}|^{2} \!+\! \sigma^{2}}\right),
\end{equation}
\noindent where $\beta$ is the transmission bandwidth. In the case where the $j$th user is served by a terrestrial BS $i$ (i.e., $ i\in \mathcal{I}-\{0\}$), the data rate of user $j$ is as in (\ref{R_RF}), which can be written as:
\begin{equation}
\label{EHMb}
R_{ij}^{Ground\_BS}=R_{ij}^{RF}, \ \forall i=1,2,\cdot\cdot\cdot N_{B}.
\end{equation}
In the other case where the $j$th user is served by the HAPS (i.e., $i=0$), the data rate via RF link (i.e., the rate expression in (\ref{R_RF}) for $i=0$. ) can be written as:
\begin{equation}
\label{EHMc}
R_{0j}^{HAPS\_RF}=R_{0j}^{RF}.
\end{equation}
Further, the data rate of user $j$ served by HAPS, donated by $R_{0j}^{HAPS}$, can be written as follows:
\begin{equation}
\label{R_HAPS}
  R_{0j}^{HAPS}=\min\{R_{0j}^{HAPS\_RF}, R^{FSO}\},
\end{equation}
For the ease of presentation, we provide a list of the expressions used in this paper in Table \ref{tableR}. The paper next presents the considered optimization problem, together with the algorithm devised to addresses the problem intricacies.

\section{Problem Formulation and Proposed Solution}
The paper focuses on  maximizing the network sum-rate by optimizing the user association strategy, and the beamforming vectors at both the HAPS and the ground base-stations, subject to {HAPS payload connectivity constraint,}
maximum transmit power constraints, and backhaul constraints. Let $P_{i}^{max}$ be the maximal allowable power at the HAPS and BSs, $\forall i\in \mathcal{I}$. Our optimization problem can then be mathematically written as follows:
\begin{subequations}
\label{EHM}
\begin{eqnarray}
\label{EHMa}
&\displaystyle\!\!\!\!\!\max_{\alpha_{ij},\mathbf{w}_{ij}}\!\!\!\!\!&  \!\sum_{j=1}^{N_{U}}\!\!\left(\!\!\sum_{i=1}^{N_{B}}\!\gamma_{ij}\alpha_{ij}R_{ij}^{Ground\_ BS}\!\!+\!\!\gamma_{0j}\alpha_{0j}R_{0j}^{HAPS}\!\!\right)\!\!\\
\label{EHMd}
&s.t.& (\ref{EHMb})-(\ref{EHMc}), (\ref{R_HAPS}),\\
\label{EHMe}
&&  \sum_{j=1}^{N_{U}}\gamma_{0j}\alpha_{0j} \leq {K_{0}},\\
\label{EHMf}
&& \sum_{i=0}^{N_{B}}\gamma_{ij}\alpha_{ij} \leq 1,\  \forall j \in \mathcal{U},\\
\label{EHMg}
&& \alpha_{ij}\in \{0, 1\},\  \forall i \in \mathcal{I}, \forall j \in \mathcal{U},\\
\label{EHMh}
&&\sum_{j=1}^{N_{U}}\gamma_{ij}\alpha_{ij}\mathbf{w}_{ij}^{H}\mathbf{w}_{ij}\leq P_{i}^{max},\ \forall i \in \mathcal{I}.
\end{eqnarray}
\end{subequations}

\noindent where the optimization in problem (\ref{EHM}) is jointly over the binary association variables $\alpha_{ij}$, and the continuous beamforming variables $\mathbf{w}_{ij}$, $\forall i\in \mathcal{I}$ and $j\in \mathcal{U}$.
{The objective function (\ref{EHMa}) is the sum-rate of the integrated network, which consists of the rates of the users served by ground BS and the rates of the users served by HAPS.}
Constraint (\ref{EHMb}) is the achievable data rate of a user served by the terrestrial BS, constraint (\ref{EHMc}) is the achievable data rate of a user served by HAPS via RF link, and (\ref{R_HAPS}) is the achievable data rate of a user served by HAPS subject to FSO backhaul constraint.
Constraint (\ref{EHMe}) guarantees that the {HAPS} can serve at most {$K_{0}$} users, and constraint (\ref{EHMf}) guarantees that a user can be served by one {transmitter} at most.
Finally, constraint (\ref{EHMh}) imposes maximal power constraints on both the HAPS and the ground BSs.

The problem (\ref{EHM}) is a mixed-integer non-convex optimization problem. The paper, therefore, next proposes solving the problem through an iterative approach. That is, first, for fixed beamforming vectors, the user association strategy is determined by linearizing the original problem so as to enable the utilization of
integer linear programming, followed by a generalized assignment problem-type solution. Then, for fixed user association, the beamforming vectors are found through a series of problem reformulations that enable the use of WMMSE-type solutions.
\subsection{User Association Strategy}
This part focuses on solving problem (\ref{EHM}) over the user association strategy $\alpha_{ij}$ by fixing the beamforming vectors at the HAPS and the ground BSs. By first substituting the min term of (\ref{R_HAPS}) in the objective function, we write problem (\ref{EHM}) as the following binary optimization problem:
\begin{subequations}
\label{EHM2}
\begin{eqnarray}
\label{EHM2a}
&\displaystyle \max_{\alpha_{ij}}& \sum_{i=1}^{N_{B}}\sum_{j=1}^{N_{U}}\gamma_{ij}\alpha_{ij}R_{ij}^{Ground\_BS}\nonumber\\
             & &
\!+\!\ \sum_{j=1}^{N_{U}}\gamma_{0j}\alpha_{0j}\min\{R_{0j}^{HAPS\_RF}, R^{FSO}\},\\
\label{EHM2b}
&s.t.&  (\ref{EHMb})-(\ref{EHMc}), (\ref{EHMe})-(\ref{EHMh})
\end{eqnarray}
\end{subequations}

\noindent where the optimization is over the binary variable $\alpha_{ij}$. Problem (\ref{EHM2}) remains, however, a complex non-linear discrete optimization problem, the global optimal solution of which would require an exhaustive search of exponential complexity. We next address such intricacies by first linearizing (\ref{EHM2}), and then adopting a GAP-based heuristic which proves to be an adequate numerical solution in the context of our problem formulation.
\subsubsection{Integer Linear Problem Formulation}
To linearize problem (\ref{EHM2}), we first replace the $\min\{.,.\}$ term in (\ref{EHM2}) with an auxiliary variable $t_{0j}$ given by:
\begin{equation}
\label{t_0j}
 t_{0j}=\min\{R_{0j}^{HAPS\_RF}, R^{FSO}\}.
\end{equation}
In order to decouple the variables $\alpha_{ij}$ from (\ref{SINR}) and linearize problem (\ref{EHM2}),
we add one auxiliary additional constraint as follows:
\begin{equation}
\label{CAdd1}
 0\leq\mathbf{w}_{ij}^{H}\mathbf{w}_{ij}\leq \gamma_{ij}\alpha_{ij}M_{ij},\ \forall i \in \mathcal{I},\ \forall j \in \mathcal{U},
\end{equation}
where $M_{ij}$ is a sufficiently large constant, added as an artifact for linearizing (\ref{EHM2}). Subject to such bounding constraint, we rewrite (\ref{R_RF}) as:
\begin{equation}
\label{R_RF_ILP}
\hat{R}_{ij}^{RF}=\beta\log_{2}\left(1+ \frac{|\mathbf{h}_{ij}^{H}\mathbf{w}_{ij}|^{2}}{\sum_{u=1,u\neq j}^{N_{U}}\sum_{b=0}^{N_{B}}|\mathbf{h}_{bj}^{H}\mathbf{w}_{bu}|^{2} + \sigma^{2}}\right).
\end{equation}
The problem (\ref{EHM2}) can, therefore, be reformulated as follows:
\begin{subequations}
\label{EHM3}
\begin{eqnarray}
\label{EHM3a}
&\displaystyle \max_{\alpha_{ij}}&\sum_{j=1}^{N_{U}}(\sum_{i=1}^{N_{B}}\gamma_{ij}\alpha_{ij}\hat{R}_{ij}^{RF}+\gamma_{0j}\alpha_{0j}t_{0j}),\\
\label{EHM3b}
&s.t.&  \sum_{j=1}^{N_{U}}\gamma_{0j}\alpha_{0j} \leq {K_{0}}, \\
\label{EHM3c}
&& \sum_{i=0}^{N_{B}}\gamma_{ij}\alpha_{ij} \leq 1,\  \forall j \in \mathcal{U},\\
\label{EHM3d}
&&\sum_{j=1}^{N_{U}}\gamma_{ij}\alpha_{ij}\mathbf{w}_{ij}^{H}\mathbf{w}_{ij}\leq P_{i}^{max},\ \forall i \in \mathcal{I},\\
\label{EHM3e}
&&  0\!\leq\!\mathbf{w}_{ij}^{H}\mathbf{w}_{ij}\!\leq\! \gamma_{ij}\alpha_{ij}M_{ij}, \ \forall i \!\in\! \mathcal{I},\ \forall j \!\in \! \mathcal{U},\\
\label{EHMg}
&& \alpha_{ij}\in \{0, 1\},\  \forall i \in \mathcal{I}, \forall j \in \mathcal{U}.
\end{eqnarray}
\end{subequations}
The optimization problem (\ref{EHM3}) becomes an integer linear problem (ILP), which can be solved using off-the-shelf available algorithms, e.g., \cite{douik2020tutorial, ganian2019solving}. ILP solvers, however, often provide suboptimal solutions to (\ref{EHM3}), and so we next improve upon the ILP solution by proposing an additional heuristic which exhibits appealing numerical prospects, as illustrated in the simulations section of the paper.

\subsubsection{Integer Linear Problem and Generalized Assignment Problem (ILP-GAP)}
To further improve upon the ILP-based solution proposed above, the paper goes one step beyond by proposing an additional heuristic that relies on maximizing an auxiliary interference-free function of the original objective function of the optimization problem (\ref{R_RF}). Such heuristic allows to use the ILP-based solution as an initial point to solve a generalized assignment problem of reasonable computational complexity; see \cite{douik2020mode,Murat,ross1975branch} and references therein. The simulations results of our paper later illustrate the numerical prospect of our proposed heuristic ILP-GAP scheme, as it outperforms the classical user association techniques.

More specifically, decouple the user association dependency by approximating the rate expression (\ref{R_RF}) with an interference-free term as:
\begin{equation}
\label{R_RF_GAP}
  \check{R}_{ij}^{RF}=\beta\log_{2}\left(1 + \frac{|\mathbf{h}_{ij}^{H}\mathbf{w}_{ij}|^{2}}{\sigma^{2}}\right).
\end{equation}
We now reformulate problem (\ref{EHM2}) as a GAP. More specifically, given the set of users $\mathcal{U}$  and the set of transmitters $\mathcal{I}$ (i.e., {knapsacks}), if the $j^{th}$ user associates with the $0^{th}$ knapsack (i.e., user $j$ is connected to the HAPS), the profit is {$t_{0j}$}. Otherwise, if the $j^{th}$ user associates with the $i^{th}$ knapsack ($i\neq 0$), the profit becomes {$\check{R}_{ij}^{RF}$}. Hence, the problem (\ref{EHM3}) can be reformulated as follows:
\begin{subequations}
\label{EHM4}
\begin{eqnarray}
\label{EHM4a}
&\displaystyle \max_{\alpha_{ij}}&\sum_{j=1}^{N_{U}}\left(\sum_{i=1}^{N_{B}}\gamma_{ij}\alpha_{ij}\check{R}_{ij}^{RF}+\gamma_{0j}\alpha_{0j}t_{0j}\right),\\
\label{EHM4b}
&s.t.&  \sum_{j=1}^{N_{U}}\gamma_{0j}\alpha_{0j} \leq {K_{0}},\\
\label{EHM4c}
&& \sum_{i=0}^{N_{B}}\gamma_{ij}\alpha_{ij} \leq 1,\  \forall j \in \mathcal{U},\\
\label{EHM4d}
&&\sum_{j=1}^{N_{U}}\gamma_{ij}\alpha_{ij}\mathbf{w}_{ij}^{H}\mathbf{w}_{ij}\leq P_{i}^{max},\ \forall i \in \mathcal{I},\\
\label{EHM4e}
&& \alpha_{ij}\in \{0, 1\},\  \forall i \in \mathcal{I}, \forall j \in \mathcal{U},
\end{eqnarray}
\end{subequations}

\noindent where constraint (\ref{EHM4b}) is the HAPS payload connectivity constraint, constraint (\ref{EHM4c}) guarantees that a user can be assigned to one transmitter only, and constraint (\ref{EHM4c}) denotes the power constraints at every transmitter $i \in \mathcal{I}$.

The above problem (\ref{EHM4}) can be readily cast as a GAP \cite{Murat}, which can be solved using a handful of efficient algorithms. In this paper, we utilize the branch and bound techniques for its provable performance guarantees \cite{douik2020tutorial}.
The solution of GAP is numerically manageable, yet strongly dependable on the initialization strategy \cite{douik2020mode,Murat,ross1975branch}. Our paper, therefore, adopts the solution reached by solving the ILP (\ref{EHM3}) as the initial point, owing to its good numerical prospects. We note that the variable $t_{0j}$ eventually gets updated after solving the above GAP. The steps of such iterative process, i.e., ILP, GAP and updating $t_{0j}$ (in this order), prove to be an efficient solution to solve the complicated user association problem (\ref{EHM2}) as shown in the simulations section, and are summarized in Algorithm 1 description below.
\begin{algorithm}
 \caption{ Determine the user association}
 \label{G1}
 \begin{enumerate}
  \item Fixed the beamforming vector ($\mathbf{w}_{ij}$).
  \item Solve the integer linear problem (\ref{EHM3}) and update $\alpha_{ij}$ to get the initial point.
  \item Set $m=0$.
  \item Use (\ref{t_0j}) to compute $t_{0j}^{0}$ and calculate the corresponding sum-rate $R_{initial}$.
  \item Define $R_{optimization}=R_{initial}$.
  \item Set $m=m+1$.
  \item Solve the general assignment problem (\ref{EHM4}) and update $\alpha_{ij}^{m}$.
  \item Calculate the sum rate $R_{sum}^{m}$ and $t_{0j}^{m}$. If $R_{sum}^{m}>R_{optimization}$, then $R_{optimization}=R_{sum}^{m}$.
  \item Go to step 7 and stop at convergence (i.e, when $|R_{sum}^{m}-R_{sum}^{m-1}|\leq \epsilon$ ).
 \end{enumerate}
 \end{algorithm}
\vspace{-0.4cm}
\subsection{Beamforming Vectors Optimization}
We now focus on finding the beamforming vectors by fixing the user association variables $\alpha_{ij}$, which are determined in the previous subsection. Problem (\ref{EHM}) can now be rewritten as:
\begin{subequations}
\label{EHM5}
\begin{eqnarray}
\label{EHM5a}
&\displaystyle\max_{\mathbf{w}_{ij}}&  \sum_{i=1}^{N_{B}}\sum_{j=1}^{N_{U}}\gamma_{ij}\alpha_{ij}R_{ij}^{Ground\_BS} \nonumber\\
             & &
\!+\!\sum_{j=1}^{N_{U}}\gamma_{0j}\alpha_{0j}\min\{R_{0j}^{HAPS\_RF},R^{FSO}\},\\
\label{EHM5b}
&s.t.&  \sum_{j=1}^{N_{U}}\gamma_{ij}\alpha_{ij}\mathbf{w}_{ij}^{H}\mathbf{w}_{ij}\leq P_{i}^{max},\ \forall i \in \mathcal{I},
\end{eqnarray}
\end{subequations}
\noindent where the optimization is over the beamforming vectors $\mathbf{w}_{ij}$. The above problem (\ref{EHM5}) is a non-convex optimization problem due the cross-mode cross-layered interference coupling in the SINR's expressions, as well as the min term stemming from the FSO backhaul constraints. The paper next tackles the difficulties of problem (\ref{EHM5}) by proposing a tweaked version of WMMSE \cite{shi2011iteratively} that best accounts for the current problem physical constraints.
\subsubsection{WMMSE Reformulation}
We first note that the minimum term in the optimization objective in (\ref{EHM5}) makes our problem different from the classical WMMSE formulation \cite{shi2011iteratively}. We, therefore, next provide a series of problem reformulations with proper outer loops updates, so as to develop a WMMSE-like solution for solving problem (\ref{EHM5}).
First, based on the values of $\alpha_{ij}$ determined in the previous subsection, one can readily determine the set of users served both the HAPS ($i=0$), and the set of users served both the ground BSs ($i=1,\cdot\cdot\cdot,N_B$). To this end, we define  $\mathcal{U}_{i}=\{j \in \mathcal{U}\ | \ \gamma_{ij}\alpha_{ij}=1\}$ as the set of users served by {transmitter} $i$ {($i=0$ for HAPS, $i=1,2 \cdot\cdot\cdot N_{B}$ for BSs)}.
Problem (\ref{EHM5}) can now be reformulated as:
\begin{subequations}
\label{EHM6}
\begin{eqnarray}
&\displaystyle \max_{\mathbf{w}_{ij}}& \sum_{i\in \mathcal{I}-\{0\}}\sum_{j\in\mathcal{U}_{i}}R_{ij}^{Ground\_BS} \nonumber\\
             & &
+\sum_{j\in\mathcal{U}_{0}}\min\{R_{0j}^{HAPS\_RF},R^{FSO}\},\\\
\label{EHM6a}
&s.t.&  \sum_{j\in\mathcal{U}_{i}}\mathbf{w}_{ij}^{H}\mathbf{w}_{ij}\leq P_{i}^{max}, \ \forall i \in \mathcal{I}.
\end{eqnarray}
\end{subequations}

Then, introduce an auxiliary variable $\tau_{0j}$ defined as:
\begin{equation}
\label{tou_0_1}
 \tau_{0j}=\min\{R_{0j}^{HAPS\_RF}, R^{FSO}\}, \ \forall j\in\mathcal{U}_{0}. \\
\end{equation}

$\tau_{0j}$ can, therefore, be written as:
\begin{equation}
\label{tou_0_2}
\tau_{0j}=
\begin{cases}
 R^{FSO},& \ R^{FSO}\leq R_{0j}^{HAPS\_RF}. \\
R_{0j}^{HAPS\_RF},& \ R^{FSO}>R_{0j}^{HAPS\_RF}.
\end{cases}
\end{equation}
We now introduce another auxiliary variable $\lambda_{ij}$, which can be regarded as the \textit{weight} of the rate-terms of user $j$ served by the {transmitter} $i$, i.e., $R_{ij}$, within the objective function of problem (\ref{EHM5}). For $i\neq0$ (i.e., in the case of ground BSs), $\lambda_{ij}=1$.  For $i=0$ (i.e., in the case of HAPS)$, \lambda_{0j}$ can be defined as:
\begin{equation}
\label{weight}
 \lambda_{0j}=
\begin{cases}
1,& \ \tau_{0j}=R_{0j}^{HAPS\_RF}.\\
0,& \ \tau_{0j}=R^{FSO}.
\end{cases}
\end{equation}
We note that the above equation (\ref{weight}) is mainly due to the fact that if $\tau_{0j}$ is equal to the constant FSO link rate, the optimization problem would no longer depend on the value of $\mathbf{w}_{0j}$, and so we can omit such constant from the objective function. Problem (\ref{EHM6}) can now be re-written as follows:
\begin{subequations}
\label{EHM7_first}
\begin{eqnarray}
&\displaystyle \max_{\mathbf{w}_{ij}}& \sum_{i\in \mathcal{I}-\{0\}}\sum_{j\in\mathcal{U}_{i}}\lambda_{ij}R_{ij}+\sum_{j\in\mathcal{U}_{0}}\lambda_{0j}\tau_{0j},\\\
\label{EHM7a}
&s.t.&  \sum_{j\in\mathcal{U}_{i}}\mathbf{w}_{ij}^{H}\mathbf{w}_{ij}\leq P_{i}^{max}, \ \forall i \in \mathcal{I}.
\end{eqnarray}
\end{subequations}

At this stage, we note that our problem reformulation (\ref{EHM7_first}) now emulates, to some extent, a sum-rate maximization problem subject to transmit power constraints, i.e., similar to the classical WMMSE formulation \cite{shi2011iteratively}. In the context of our paper, problem (\ref{EHM6}) has the equivalent optimal solution with the following WMMSE minimization problem:
\begin{subequations}
\label{EHM7}
\begin{eqnarray}
&\displaystyle  \min_{\boldsymbol{\rho}_{ij},\mathbf{u}_{ij},\mathbf{w}_{ij}} & \sum_{i\in\mathcal{I}}\sum_{j\in\mathcal{U}_{i}}\lambda_{ij}\left(\mathrm{Tr}(\boldsymbol{\rho}_{ij}\mathbf{e}_{ij})-\log \boldsymbol{\rho}_{ij}\right),\\
\label{EHM8a}
&s.t.&  \sum_{j\in\mathcal{U}_{i}}\mathbf{w}_{ij}^{H}\mathbf{w}_{ij}\leq P_{i}^{max},\ \forall i \in \mathcal{I},
\end{eqnarray}
\end{subequations}
where $ P_{i}^{max}$ is the maximum power of BS $i$, $\boldsymbol{\rho}_{ij}$ denotes the mean squared
error (MSE) weight for user $j$ served by {transmitter} $i$ (i.e., $\forall j \in \mathcal{U}_{i}$), and $\mathbf{u}_{ij}$ is the receive beamforming vector at the user $j$ when served by {transmitter} $i$. Finally, $\mathbf{e}_{ij}$ is the MSE at the user $j$ when served by {transmitter} $i$ defined as:
\begin{footnotesize}
\begin{equation}
\begin{split}
\label{error}
  \mathbf{e}_{ij}=&(\mathbf{I}-\mathbf{u}_{ij}^{H}\mathbf{h}_{ij}^{H}\mathbf{w}_{ij})(\mathbf{I}-\mathbf{u}_{ij}^{H}\mathbf{h}_{ij}^{H}\mathbf{w}_{ij})^{H}
  \\&+\sum_{(b,l)\neq(i,j)}\mathbf{u}_{ij}\mathbf{h}_{bj}^{H}\mathbf{w}_{bl}
  \mathbf{w}_{bl}^{H}\mathbf{h}_{bj}\mathbf{u}_{ij}^{H}
  +\sigma^{2} \mathbf{u}_{ij}^{H}\mathbf{u}_{ij}, \forall i \in \mathcal{I},\ \forall j\in\mathcal{U}_{i}.
  \end{split}
\end{equation}
\end{footnotesize}
\subsubsection{Beamforming Algorithm (Algorithm 2)}
The reformulated problem (\ref{EHM7}) is convex in each of the optimization variables  $\boldsymbol{\rho}_{ij},\mathbf{u}_{ij},\mathbf{w}_{ij}$. Therefore, one can solve (\ref{EHM7}) via finding one variable by fixing two other variables. More specifically, $\forall i \in \mathcal{I}, \ \forall j\in\mathcal{U}_{i}$, the optimal receiver $\mathbf{u}_{ij}$ under fixed $\mathbf{w}_{ij}$ and $\boldsymbol{\rho}_{ij}$ is an MMSE receiver defined by:
\begin{equation}
  \label{MMSE receiver}
  \mathbf{u}_{ij}=\mathbf{u}_{ij}^{mmse}=\frac{\mathbf{h}_{ij}^{H}\mathbf{w}_{ij}}{\sum_{b\in\mathcal{I}}\sum_{l\in\mathcal{U}_{b}}\mathbf{h}_{bj}^{H}\mathbf{w}_{bl}
  \mathbf{w}_{bl}^{H}\mathbf{h}_{bj}+\sigma^{2}\mathbf{I}}.
\end{equation}
Similarly, the optimal MSE weight $\boldsymbol{\rho}_{ij}$ under fixed $\mathbf{u}_{ij}$ and $\mathbf{w}_{ij}$ can be written as:
\begin{equation}
  \label{MSE_weight}
  \boldsymbol{\rho}_{ij}=\mathbf{e}_{ij}^{-1},\forall i \in \mathcal{I}, \ \forall j\in\mathcal{U}_{i}.
\end{equation}
Lastly, finding the optimal transmit beamformer $\mathbf{w}_{ij}$ under fixed $\boldsymbol{\rho}_{ij},\mathbf{u}_{ij}$ can be cast as convex quadratic optimization problem, which can be solved efficiently \cite{boyd2004convex}. The above updates of $\boldsymbol{\rho}_{ij},\mathbf{u}_{ij},\mathbf{w}_{ij}$, i.e., (\ref{error}-\ref{MSE_weight}),  are eventually executed in an iterative way together with the proper updates of $\tau_{ij}$ and $\lambda_{ij}$ according to (\ref{tou_0_2}) and (\ref{weight}), respectively, which enables finding the beamforming vectors $\mathbf{w}_{ij}$ efficiently, as presented in Algorithm 2. Such algorithm is in fact guaranteed to converge to a stationary point of (\ref{EHM6}), as further illustrated in the next lemma.
\begin{algorithm}[h!]
 \caption{Determine Beamforming Vectors}
 \label{G2}
 \begin{enumerate}
  \item Fix the user association strategy.
  \item set $m=0$.
  \item Fix initial beamforming vectors $\mathbf{w}_{ij}$.
  \item Calculate the $\tau_{ij}^{m}$ and determine the $\lambda_{ij}^{m}$, according to (\ref{tou_0_2}), (\ref{weight}).
  \item Fix $\mathbf{w}_{ij}$, and update $\mathbf{u}_{ij}$, according to (\ref{MMSE receiver}).
  \item Fix $\mathbf{u}_{ij}$ and $\mathbf{w}_{ij}$, and update $\boldsymbol{\rho}_{ij}=\mathbf{e}_{ij}^{-1}$.
  \item Calculate and update the optimal transmit beamformer $\mathbf{w}_{ij}$ under fixing $\boldsymbol{\rho}_{ij},\mathbf{u}_{ij}$.
  \item Compute the sum-rate $R_{sum}^{m}$.
  \item set $m=m+1$.
  \item Go to step 4 and stop at convergence (i.e, when $|R_{sum}^{m}-R_{sum}^{m-1}|\leq \epsilon$ ).
 \end{enumerate}
 \end{algorithm}

\begin{lemma}
The solution obtained by Algorithm 2 converges to a stationary point of (\ref{EHM6}).
\end{lemma}

\begin{proof}
The steps of the proof of Lemma 1 are included in Appendix \ref{ap1} of the paper.
\end{proof}

\subsection{Overall Algorithm and Convergence}
Now that both the discrete and continuous variables of problem (\ref{EHM}) are determined, as per Algorithms 1 and 2, respectively, the paper adopts an iterative algorithm to optimize both variables alternatively. Specifically, the solution involves three loops: two inner loops and one outer loop. The first inner loop solves the user association strategy, and the second inner loop updates the beamforming vectors at HAPS and ground BSs. The outer loop, finally, combines two inner loops to optimize the user association and beamforming. Since each of the two loops provides a nondecreasing function in the network sum-rate (which is bounded by the network capacity), the overall algorithm is guaranteed to converge, as also validated later through the paper simulations. The steps of the overall algorithm are shown in Algorithm \ref{G3} description below.
\begin{algorithm}[h!]
 \caption{Overall Algorithm}
 \label{G3}
 \begin{enumerate}
  \item Generate initial beamforming vectors $(\mathbf{w}_{ij})$.
  \item Repeat
  \item Fix the beamforming vector of BSs and HAPS.
  \item Implement Algorithm 1 to update the $\alpha_{ij}$.
  \item Implement Algorithm 2 to update the beamforming vectors of all users.
  \item Compute the sum-rates $R_{sum}$ of the network.
  \item Stop at convergence.
 \end{enumerate}
 \end{algorithm}
\subsection{Computational Complexity}
To best characterize the computational complexity of the proposed algorithm, we note the Algorithm \ref{G3} solves two problems sequentially. The user association problem can be solved by an integer linear program and a GAP-based algorithm. The linear programming has a polynomial-time solvable computational complexity $O(n^{k})$, where $n$ is the number of variables $\alpha_{ij}$ and $k$ is the degree of complexity. Our paper adopts a GAP-solver based on branch-and-bound (BnB) solutions \cite{douik2020tutorial}, and so the GAP step complexity is in the order of $O(\xi^{n})$, where $1< \xi <2$. The beamforming solution, on the other hand, relies on WMMSE \cite{shi2011iteratively}, the per-iteration computational complexity of which is upper-bounded by $O(N_{U}^{2}N_{A}+N_{U}^{2}N_{A}^{2}+N_{U}^{2}N_{A}^{3}+N_{U})$, where $N_{A}$ is the maximum number of antennas across both HAPS and ground BSs.

To summarize, our proposed algorithm includes three loops. When the numbers of iterations of the two inner loops are $T_{1}$ and $T_{2}$, respectively, and that of the outer loop is $T_{3}$, the computational complexity of the overall algorithm becomes $T_{3}[O(n^{k})+T_{1}O(\xi^{n})+T_{2}O(N_{U}^{2}N_{A}+N_{U}^{2}N_{A}^{2}+N_{U}^{2}N_{A}^{3}+N_{U})]$, which is reasonably dependable on the particular GAP solution complexity.

\subsection{Baseline Approaches}
As mentioned above, the major complexity of the proposed solution originates from the user association strategy, especially the GAP algorithm. To this end, the paper now presents two alternative low complexity methods that depend on the distance and channel values, respectively. Similar approach can be found in \cite{dahrouj2015distributed}.
\subsubsection{Baseline 1 (Distance dependent approach)}
This method assigns user $j$ to {transmitter} $i$ ($\forall i\in \mathcal{I}$ and $j\in \mathcal{U}$) based on their mutual distance, denoted by $d_{ij}$. Let $\mathbf{D}$ be the $(N_{B}+1) \times N_{U}$ matrix whose entries, $d_{ij}$, denote the distance between the {transmitter} $i$ and the user $j$, i.e., the $(i,j)^{th}$ entry of the matrix $\mathbf{D}$ is $\mathbf{D}_{i,j}=d_{ij}$. At each step, find the smallest entry of matrix $\mathbf{D}$, call it $\mathbf{D}_{i^{min},j^{min}}$. User $j^{min}$ then maps to {transmitter} $i^{min}$, as long as
{each ground BS does not serve more than its number of antennas and that the HAPS payload connectivity constraint is satisfied.
} Next, delete the $(\mathbf{D}_{j^{min}})^{th}$ column of the matrix, so that user $j^{min}$ cannot be associated with other {transmitters} in subsequent steps. Repeat the above procedure until all users are connected to {transmitters} or all {transmitters}
 {resource constraints (\ref{EHMe}), (\ref{EHMf}) are violated.}
\subsubsection{Baseline 2 (Channel dependent approach)}
Unlike the distance dependent approach,
this method assigns users to the {transmitters} based on the channel gain between the {transmitters} and the users, denoted by $c_{ij}=||\mathbf{h}_{ij}||_2^2$, $\forall i\in \mathcal{I}$ and $j\in \mathcal{U}$. Let $\mathbf{C}$ be the $(N_{B}+1) \times N_{U}$ matrix whose entries are the channel gains between {transmitter} $i$ and user $j$ denoted by $c_{ij}$, i.e., the $(i,j)^{th}$ entry of the matrix $\mathbf{C}$ is $\mathbf{C}_{i,j}=c_{ij}$.
At each step, find the largest entry of matrix $\mathbf{C}$, call it $\mathbf{C}_{i^{max},j^{max}}$. User $j^{max}$ then maps to {transmitter} $i^{max}$, as long as the resource constraints of {transmitter} $i^{max}$ are satisfied. Next, delete the $(\mathbf{C}_{j^{max}})^{th}$  column of the matrix, so that user $j^{max}$ cannot be associated with other {transmitters} in subsequent steps. The procedure then gets repeated as in the distance dependent approach above.

\section{Simulation Results}
This section evaluates the performance of the proposed 
algorithm for various networks scenarios, so as to illustrate the numerical gains of the developed joint user association and beamforming optimization framework in the context of integrated satellite-HAPS-ground networks. The paper particularly compares the proposed joint optimization solution adopting ILP and GAP (IG) as user association strategy and WMMSE as beamforming in high backhaul capacity (HBC-IG-WMMSE) to 8 different benchmarks:
1- joint optimization with channel-dependent (CD) and WMMSE in high backhaul capacity (HBC-CD-WMMSE),
2- joint optimization with distance-dependent (DD) and WMMSE in high backhaul capacity (HBC-DD-WMMSE),
3- joint optimization solution {with ILP-GAP and WMMSE} in low backhaul capacity (LBC-IG-WMMSE),
4- joint optimization with channel-dependent and WMMSE in low backhaul capacity (LBC-CD-WMMSE),
5- joint optimization with distance-dependent and WMMSE in low backhaul capacity (LBC-DD-WMMSE),
6- ILP-GAP approach in high backhaul capacity (HBC-IG),
7- channel-dependent approach in high backhaul capacity (HBC-CD),
and
8- distance-dependent approach in high backhaul capacity (HBC-DD).
For completeness, we also summarize the above algorithms in Table \ref{table_II}.

\begin{table*}[!t]
\centering
\caption{Algorithms Abbreviation}
\label{table_II}
\begin{tabular}{|p{.25\textwidth} | p{.7\textwidth} | }
\hline
  \textbf{Algorithm} &   \textbf{Definition} \\
 \hline
 HBC-IG-WMMSE &  Joint optimization with ILP-GAP (IG) and WMMSE in high backhaul capacity\\
  \hline
HBC-CD-WMMSE & Joint optimization with channel-dependent (CD) and WMMSE in high backhaul capacity \\
 \hline
HBC-DD-WMMSE & Joint optimization with distance-dependent (CD) and WMMSE in high backhaul capacity \\
 \hline
LBC-IG-WMMSE &Joint optimization with ILP-GAP (IG) and WMMSE in low backhaul capacity \\
 \hline
LBC-CD-WMMSE & Joint optimization with channel-dependent (CD) and WMMSE in low backhaul capacity \\
 \hline
LBC-DD-WMMSE & Joint optimization with distance-dependent (CD) and WMMSE in low backhaul capacity \\
 \hline
HBC-IG & ILP-GAP (IG) approach in high backhaul capacity\\
\hline
HBC-CD & Channel-dependent (CD) approach in high backhaul capacity \\
 \hline
HBC-DD & Distance-dependent (DD) approach in high backhaul capacity \\
\hline
\end{tabular}
\end{table*}


We first simulate a network of medium size with ground footprint
of $5$ km $\times$ $5$ km, similar to Fig. \ref{SM}. We herein assume that $N_{U}$ users are distributed in two different subareas. Subarea $1$ contains $12$ BSs with coordinates: x: ($0$ km to $1$ km) and y: ($0$ km to $1$ km) and contains $60\%$ of the total number of users. The remaining area is the subarea $2$ and contains $40\%$ of the total number of users, with no deployed BS. In this case, subarea 1 can be considered as an urban area, while subarea 2 can be considered as a suburban area. The satellite is fixed at the coordinates $[2.5, 2.5, 36000]$ km. The HAPS is also fixed at the coordinates $[2.5, 2.5, 18]$ km. Table \ref{tableIII} presents the values of the
parameters used in the simulations (unless mentioned otherwise); the FSO-related parameters are adopted from \cite{alzenad2018fso}. For illustration purposes, the data-availability variables $\gamma_{ij}$ are set to 1 throughout the simulations section. {Further, for illustration purposes, we set the maximum number of users that the HAPS can serve to the number of its corresponding antennas (i.e., $K_{0}=N_{A}^{0}$).}

\begin{table*}[!t]
\centering
\caption{Simulation Parameters}
\label{tableIII}
\begin{tabular}{|p{.45\textwidth} | p{.2\textwidth} |}
\hline
Parameter Name& Parameter Value\\
\hline
  The height of HAPS, $z^{HAPS}$ & $18$ km  \\
  The height of geo-satellite, $z^{satellite}$ & $36000$ km\\
  Bandwidth of BSs, $\beta$ & $10$ MHz  \\
  Central carry frequency, $f_{c}$ & $3$ GHz \\
  Rician factor, $\kappa_{HAPS}$ & $5$\  \\
  Noise power, $N_{0}$& $-174$ dBm/Hz\  \\
  Ground base-station antenna, $N_{A}^{BS}$& $1$ \\
  HAPS antennas in medium network, $N_{A}^{HAPS\_mid}$& $20$\\
  HAPS antennas in large network, $N_{A}^{HAPS\_big}$& $40$\\
  Maximum power of urban BS, $P_{BS}^{max\_urban}$& $1$ watt\  \\
  Maximum power of suburban BS, $P_{BS}^{max\_suburban}$& $2$ watt\  \\
  Maximum power of rural BS, $P_{BS}^{max\_rural}$& $5$ watt\  \\
  Maximum power of HAPS in mid network, $P_{HAPS}^{max\_mid}$& $100$ watt\\
  Maximum power of HAPS in big network, $P_{HAPS}^{max\_big}$& $200$ watt\\
  Standard deviation of ground-level shadowing $\sigma_{a}$& $5$dB \ \\
  \hline
\end{tabular}
\end{table*}

\begin{figure}[!t]
\centering
\includegraphics[width=3in]{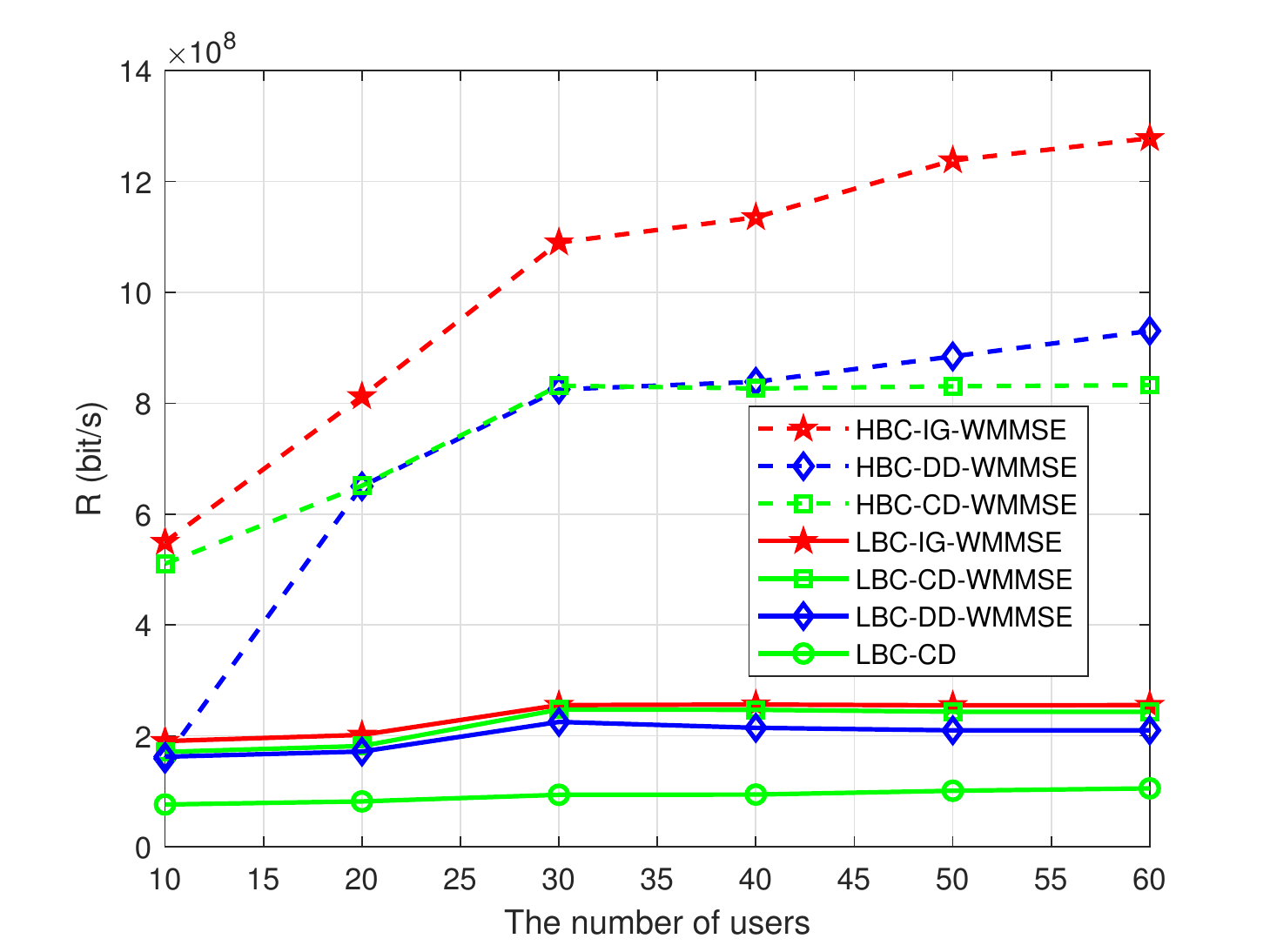}
\caption{Sum-rate versus the total number of users.}
\label{miduser}
\end{figure}

\begin{figure}[!t]
\centering
\includegraphics[width=3in]{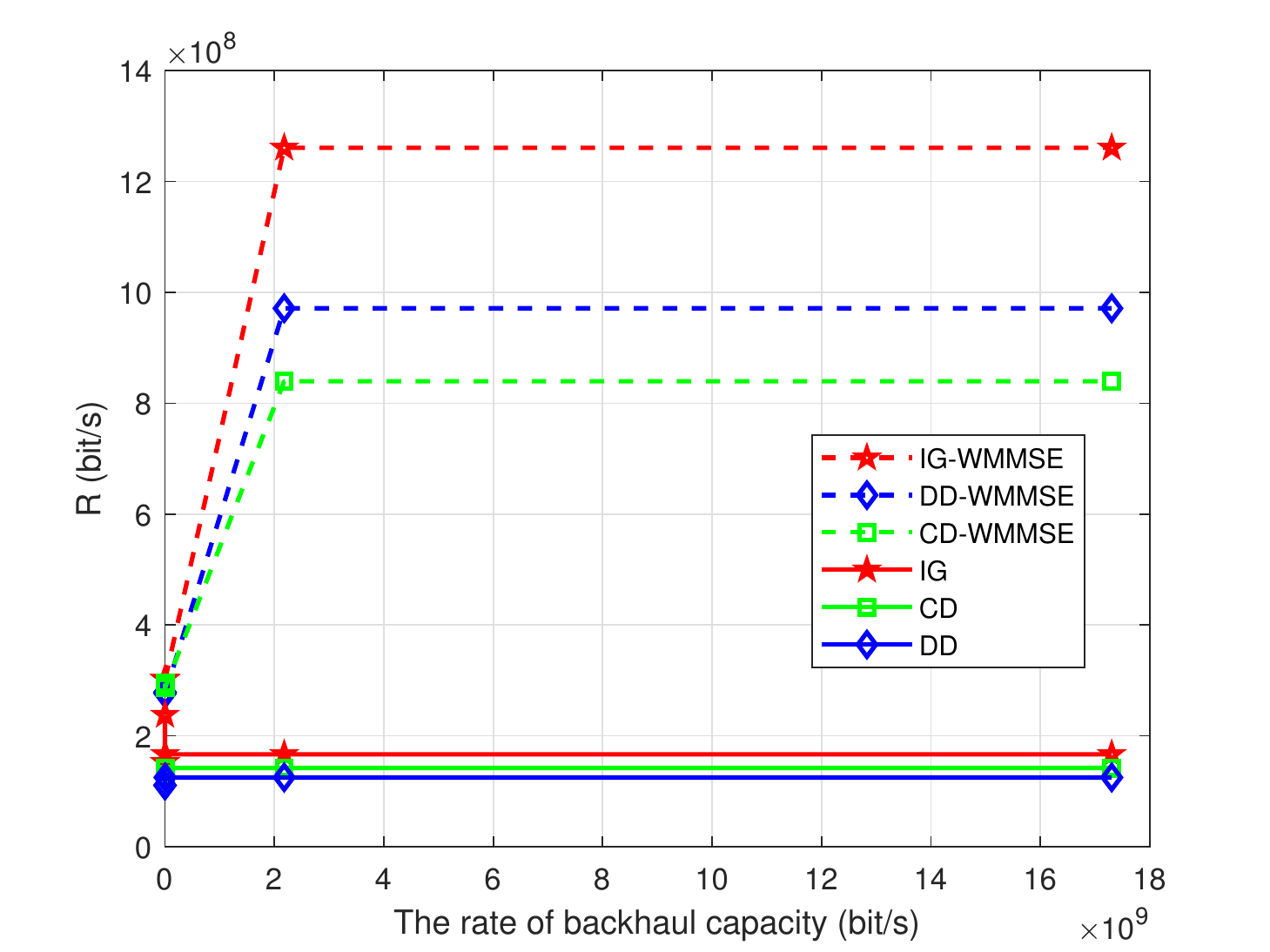}
\caption{Sum-rate versus the backhaul capacity.}
\label{midfso}
\end{figure}

We first illustrate the impact of the number of users on the network performance in Fig. \ref{miduser}, which plots the sum-rate versus the total number of users. Fig. \ref{miduser} shows how the proposed solution provides a substantial gain compared to other algorithms, especially when the number of users increases. This is particularly the case since the interference level becomes larger in denser networks, and so the role of the proposed resource allocation scheme in mitigating interference becomes more pronounced.
Fig. \ref{miduser} also shows how the joint optimization attains a considerable improvement at high backhaul capacity.
This is due to the fact that, unlike low backhaul capacity regimes, high backhaul capacity regimes match a fully enabled HAPS, which unleashes the full power of the HAPS towards serving more users, thereby increasing the network total throughput.

Judging from Fig. \ref{miduser}, when the number of users is $50$, one can notice that the network resources are fully utilized. Therefore, we now use $N_{U}=50$
so as to better characterize the impact
of backhaul capacity on the network sum-rate performance, by plotting the sum-rate versus the backhaul capacity in Fig. \ref{midfso}. The figure shows that our proposed solution always attains the highest sum-rate as compared to all other classical strategies. The figure also illustrates the significant gap between joint optimization and user association strategy, which highlights the importance of the beamforming step at mitigating the interference, i.e., beyond the initial user association step.
{When the backhaul capacity exceeds the data rate of the RF link, the $\min\{.,.\}$ term in (\ref{R_HAPS}) becomes equal to the data rate of the RF link. This is the reason why the sum-rate becomes invariant at the high backhaul regime. }
We also can notice that the proposed user association strategy does well in the lowest backhaul capacity. Due to the FSO backhaul constraint (\ref{R_HAPS}), when the data rate of the FSO is as low as 0, the user can reasonably not choose the HAPS to connect to, which is depicted through the behavior of the ILP-GAP algorithm. However, the users association strategies of the two baseline algorithms, i.e., DD and CD, are only based on distance and channel gains, and so users may still choose to associate to HAPS at low backhaul capacity, which introduces high interference to the RF network and, at the same time, exacerbates the network sum-rate.

\begin{figure}[!t]
\centering
\includegraphics[width=3in]{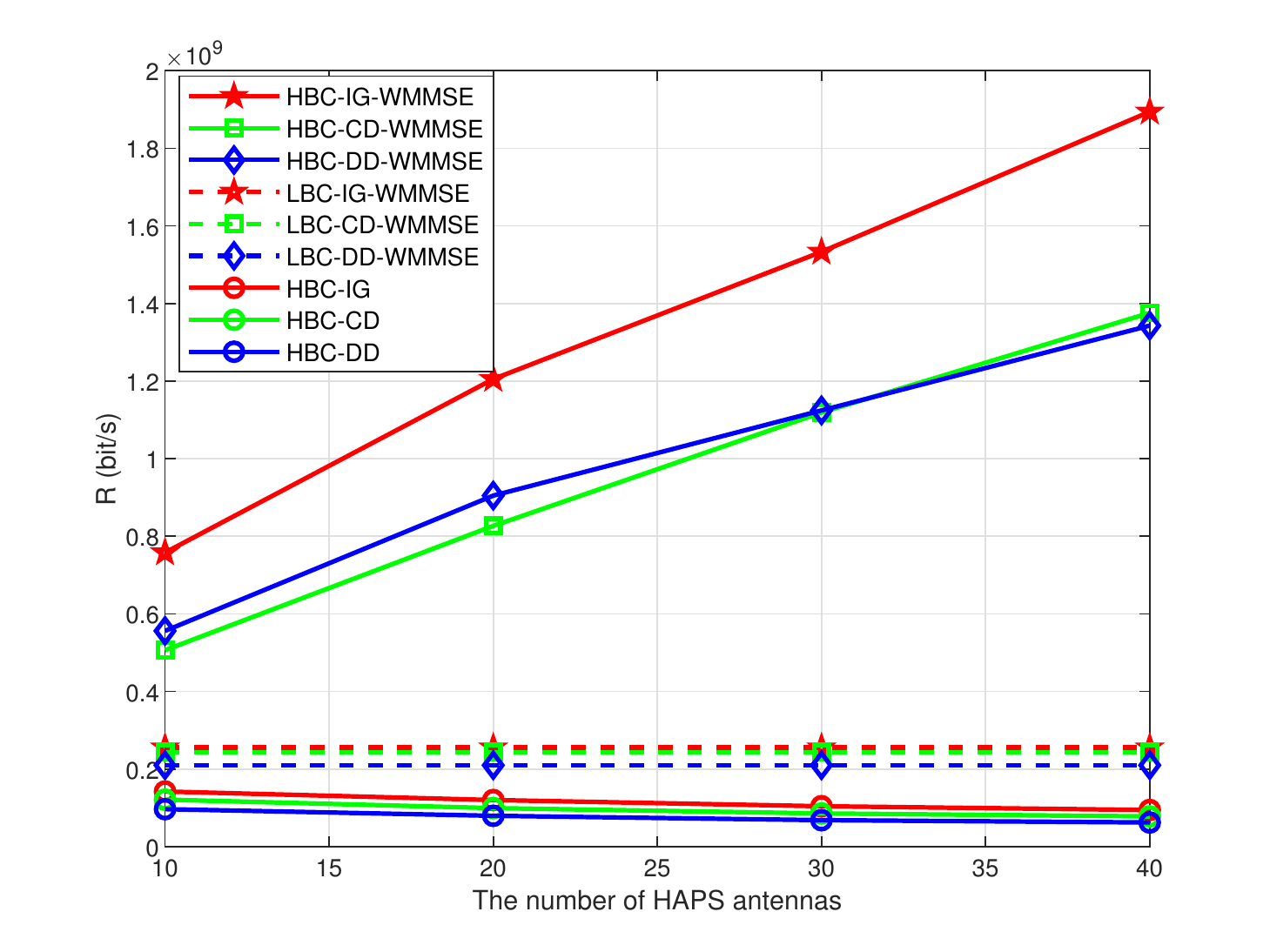}
\caption{Sum-rate versus the total number of HAPS antennas.}
\label{midAN}
\end{figure}

To illustrate the impact of the HAPS antennas on the network perfromance, Fig. \ref{midAN} plots the sum-rate versus the number of HAPS antennas {with $N_{U}=50$}. 
The figure shows that,
as the number of antennas increases, {the sum-rate resulting from} the solutions that rely on user association only slightly decreases, while {the sum-rate resulting from} the solutions which implement the additional beamforming optimization step increases. This is because if the HAPS has more antennas, the HAPS can serve more users, which introduces more interference. Since user association strategy can not alone reduce the high level of interference stemming from the HAPS newly deployed antennas, the sum-rate slightly decreases. However, the additional beamforming optimization step can significantly mitigate the interference due to the empowered spatial multiplexing capabilities.
Fig. \ref{midAN} further shows that the proposed joint approach always outperforms all other baseline solutions for all the simulated scenarios. The figure particularly shows the gain harvested through augmenting the ground networks with HAPS capabilities, which is shown through the substantial gain at high backhaul capacity (i.e., when the rate of the HAPS to ground users RF link is inferior to the FSO link capacity) as compared to the low backhaul capacity (i.e., when the HAPS potential is rather limited by the FSO link).
In fact, at the low backhaul capacity, increasing the number of antennas does not change the sum-rate as the HAPS remains idle in this case. On the opposite, at the high backhaul capacity regime, the active operation of the HAPS becomes a major driver in pushing the network throughput upward.

\begin{figure}[!t]
\centering
\includegraphics[width=3in]{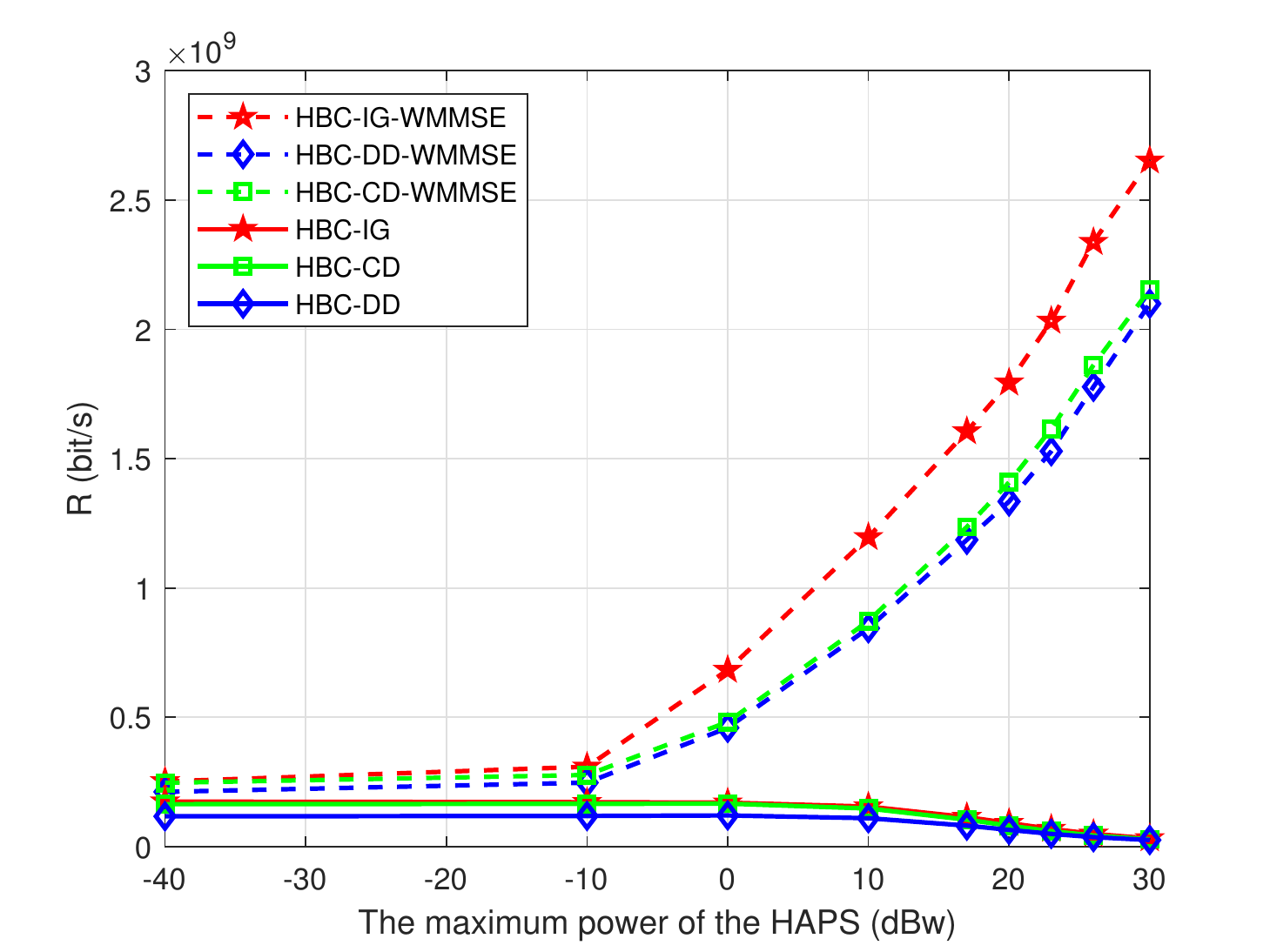}
\caption{Sum-rate versus maximum power of HAPS.}
\label{midHAPS}
\end{figure}

Fig. \ref{midHAPS} shows the sum-rate of the network versus the maximum power of the HAPS, when the number of HAPS antennas is set to $N_{A}^{0}=40$ {and the number of users is $50$}. Not only does Fig. \ref{midHAPS} reaffirm the superiority of our proposed joint optimization, but the figure also illustrates how the additional power capability at the HAPS helps increasing the network sum-rate.
This is because when the maximum power of HAPS increases, the number of users served by the HAPS increases. Given the strong capability of the proposed scheme at mitigating the cross-mode cross-layered interference, the sum-rate does indeed increase with the joint optimization scheme. Such result is further highlighted by depicting the fraction of users served by the HAPS out of the total number of users (denoted by $\delta$) in the high backhaul capacity as shown in Fig. \ref{delta}. The figure shows that as the maximum power of HAPS increases, the number of users served by HAPS increases, which reflects the ability of higher power HAPS to serve more users.
The figure further shows that the greater the ground-level shadowing is, the more users the HAPS would serve. This is because when the ground-level shadowing increases, the gain brought by the base-station connecting the user to the network decreases, and so the users are more inclined to be served by the HAPS.
Likewise, if the HAPS is equipped with more antennas, more users tend to be served by the HAPS, which explains the capacity boost when the HAPS has $40$ antennas. Fig. \ref{delta} is indeed a crisp illustration of how HAPS help serving users in both urban and suburban areas; thereby highlighting HAPS roles in connecting the unconnected (through strong HAPS capabilities), and super-connecting the connected (at higher interference levels).

\begin{figure}[!t]
\centering
\includegraphics[width=3in]{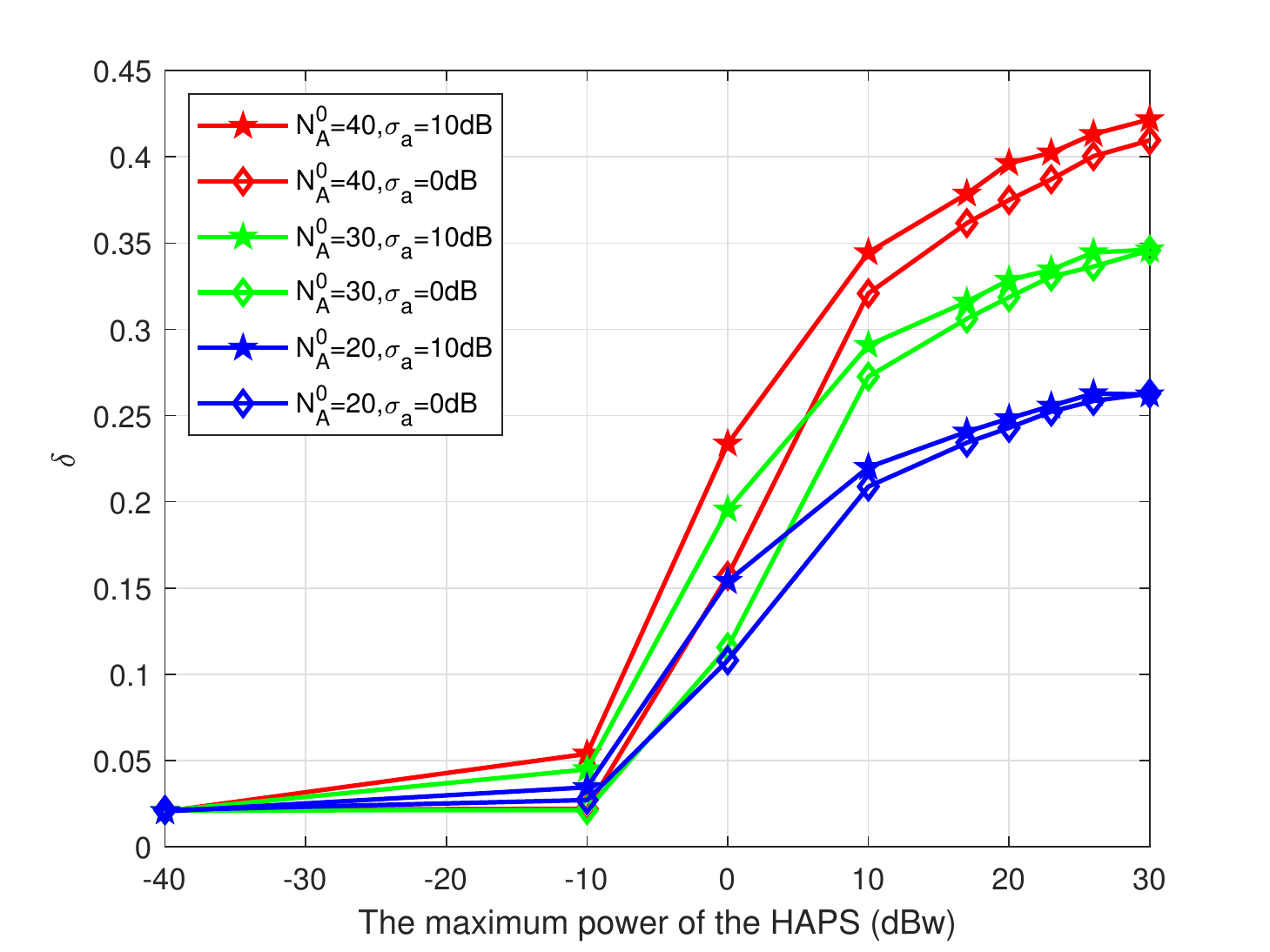}
\caption{Delta versus maximum power of HAPS.}
\label{delta}
\end{figure}

\begin{figure}[!t]
\centering
\includegraphics[width=3in]{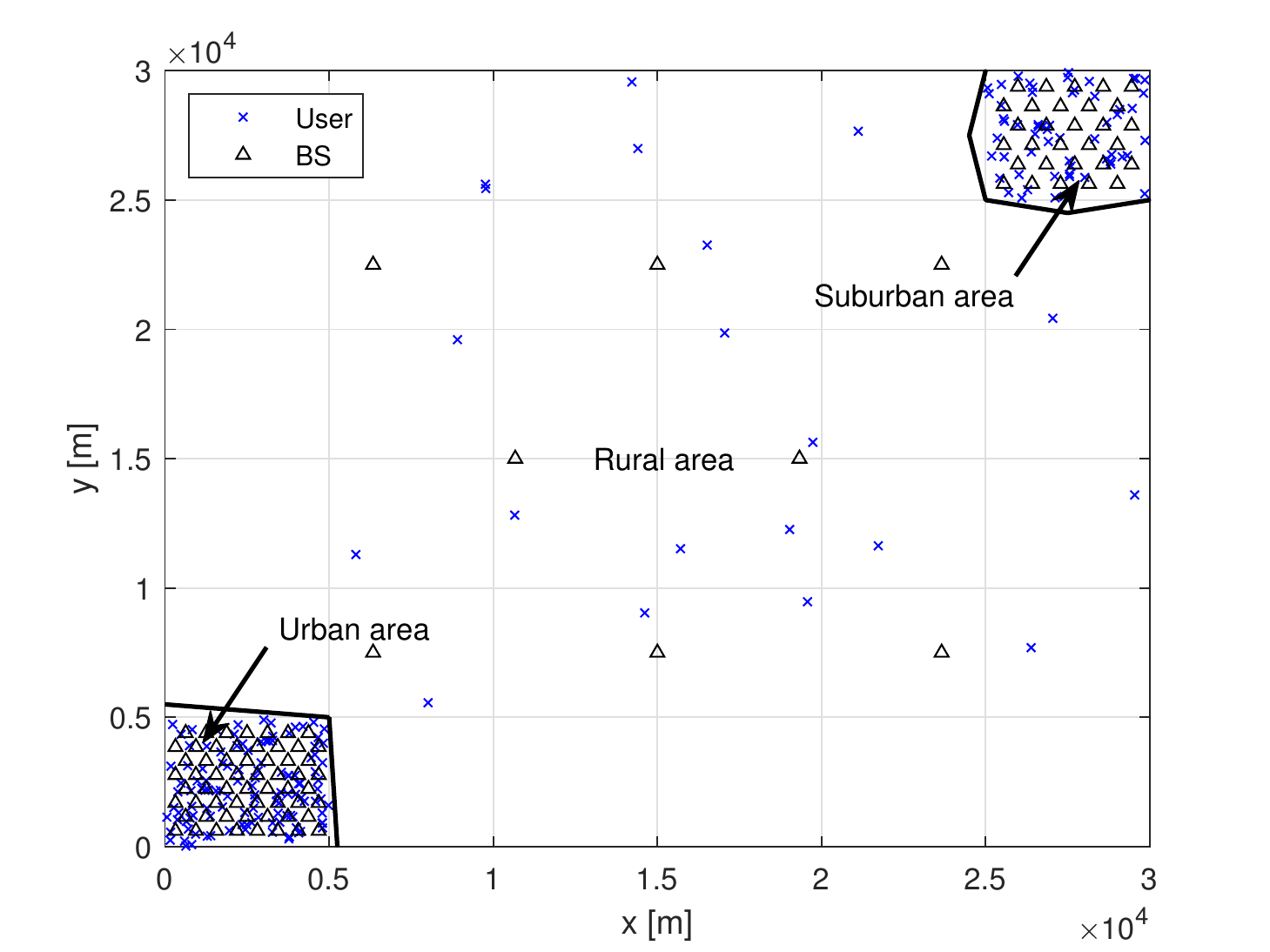}
\caption{The layout of the big network.}
\label{bignet}
\end{figure}

\begin{figure}[!t]
\centering
\includegraphics[width=3in]{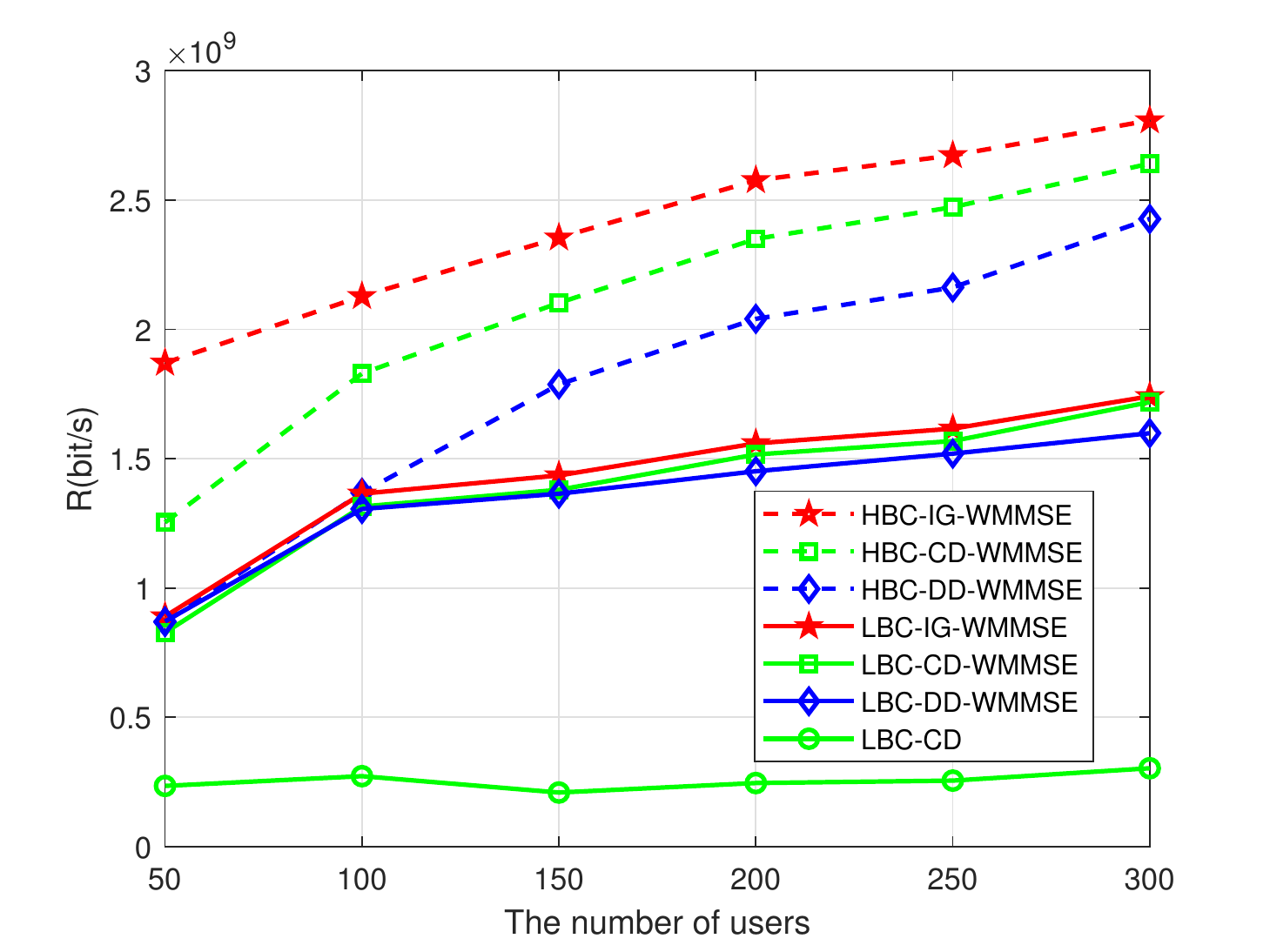}
\caption{Sum-rate versus users.}
\label{biguser}
\end{figure}

We next simulate a network of large size with a ground footprint of $30$ km $\times$ $30$ km as shown in Fig. \ref{bignet}. In this area, $N_{U}$
users are distributed in three different subareas. Subarea $1$ contains 60 BSs with coordinates: x: ($0$ km to $5$ km) and y: ($0$ km to $5$ km) and contains $60\%$ of the total number of users. Subarea $2$ contains 30 BSs with coordinates: x: ($25$ km to $30$ km) and y: ($25$ km to $30$ km) and contains $30\%$ of the total number of users. The remaining area is the subarea $3$ and contains $10\%$ of the total number of users and $8$ BSs. In this case, subarea $1$ can be considered as an urban area. While subarea $2$ can be considered as a suburban area, subarea $3$ can be considered as a rural area.
The geo-satellite and HAPS locations remain {at the center of the network} as before.


Fig. \ref{biguser} shows the sum-rate versus the total number of users. It is observed that the proposed solution outperforms all other approaches, especially when the number of users increases. Fig. \ref{biguser} also shows that joint optimization methods are superior to those algorithms with only user association. Note that the utility of the network in the high backhaul capacity is better than the low backhaul capacity, which further indicates the positive impact of HAPS on large networks throughput.
\begin{figure}[!t]
\centering
\includegraphics[width=3in]{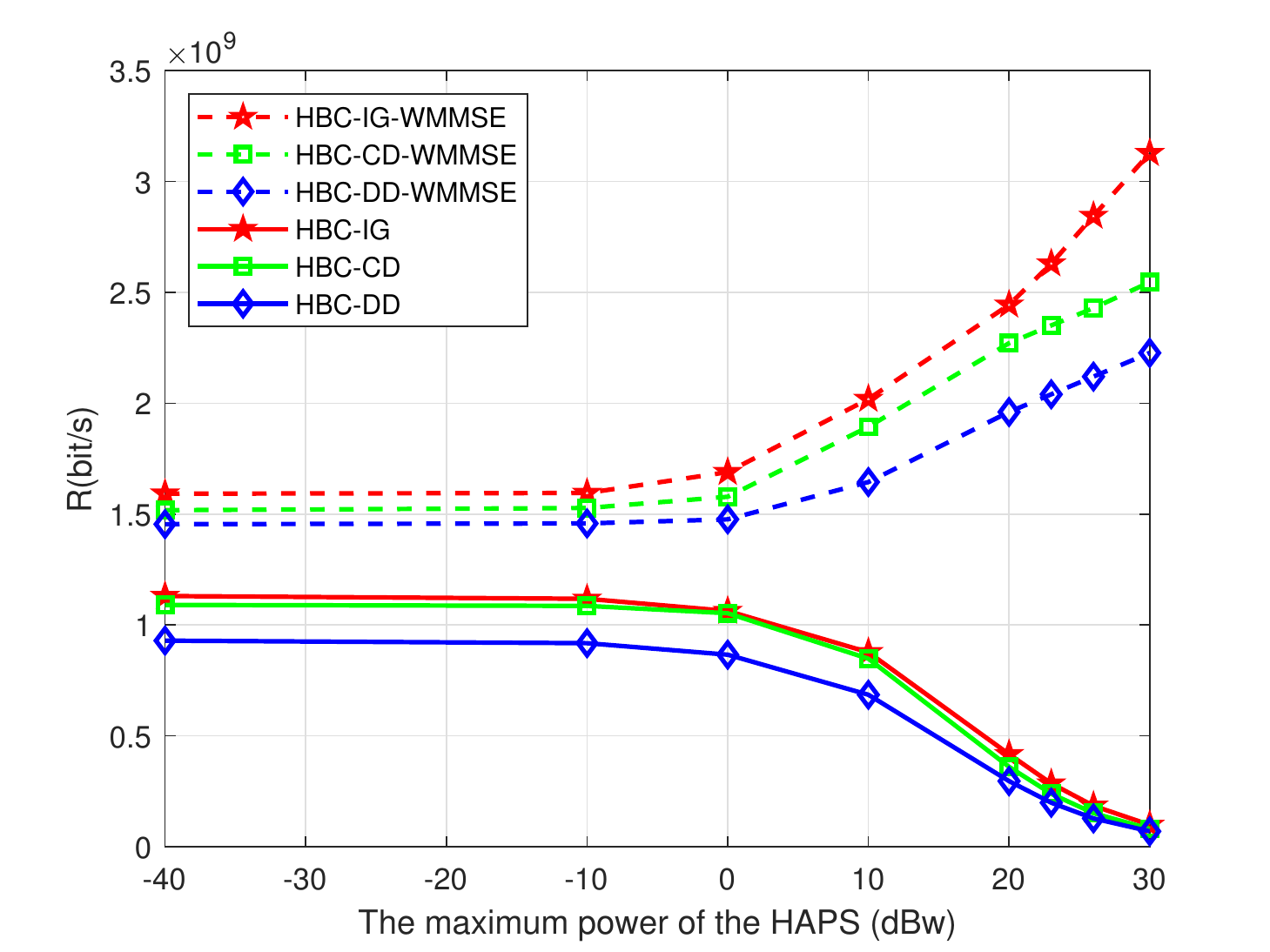}
\caption{Sum-rate versus maximum power of HAPS.}
\label{bigHAPS}
\end{figure}
Fig. \ref{bigHAPS} shows the sum-rate versus the maximum power of HAPS {with $N_{U}=200$}.
Similar to Fig. \ref{midHAPS}, Fig. \ref{bigHAPS} shows that when the maximum power of the HAPS increases, {the sum-rate resulting from} the joint optimization increases, while {the sum-rate resulting from} user association only decreases. It is particularly noticeable how the proposed algorithm can bring more significant improvement to large networks than medium networks, mainly due to the higher interference levels. In fact, when the maximum power of HAPS is $30$dBw, Fig. \ref{midHAPS} shows that the proposed algorithm can improve the network performance by $22.7\%$. For large networks, the network performance can be improved by $25.5\%$, which indicates that the proposed algorithm can improve the network performance, particularly in ultra-dense networks.
\begin{figure}[!t]
\begin{center}
\subfigure[Maximum power of HAPS is -40dBw]{
\includegraphics[width=2.5in]{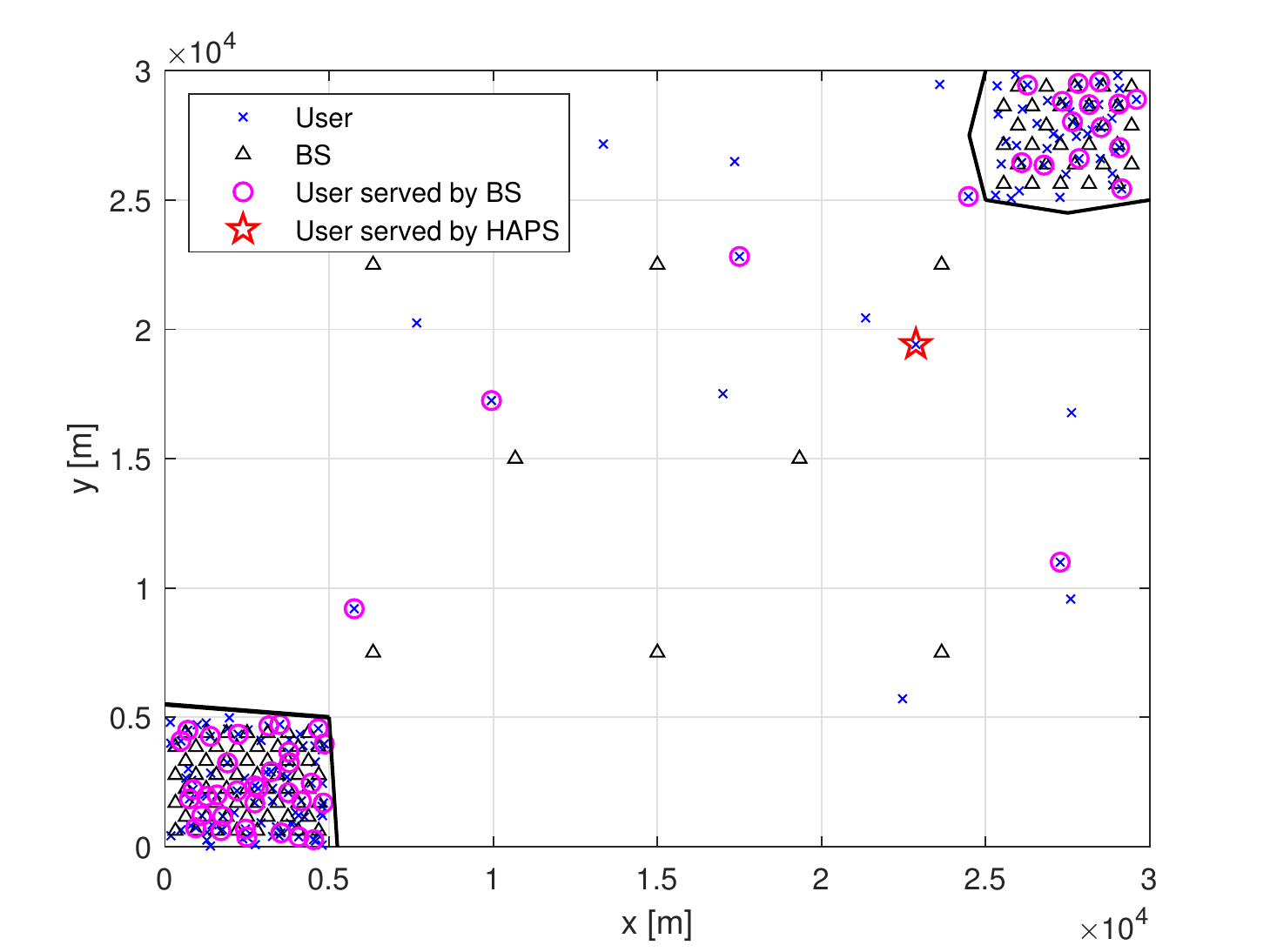}
\label{power_-40dBw}
}
\subfigure[Maximum power of HAPS is 10dBw]{
\includegraphics[width=2.5in]{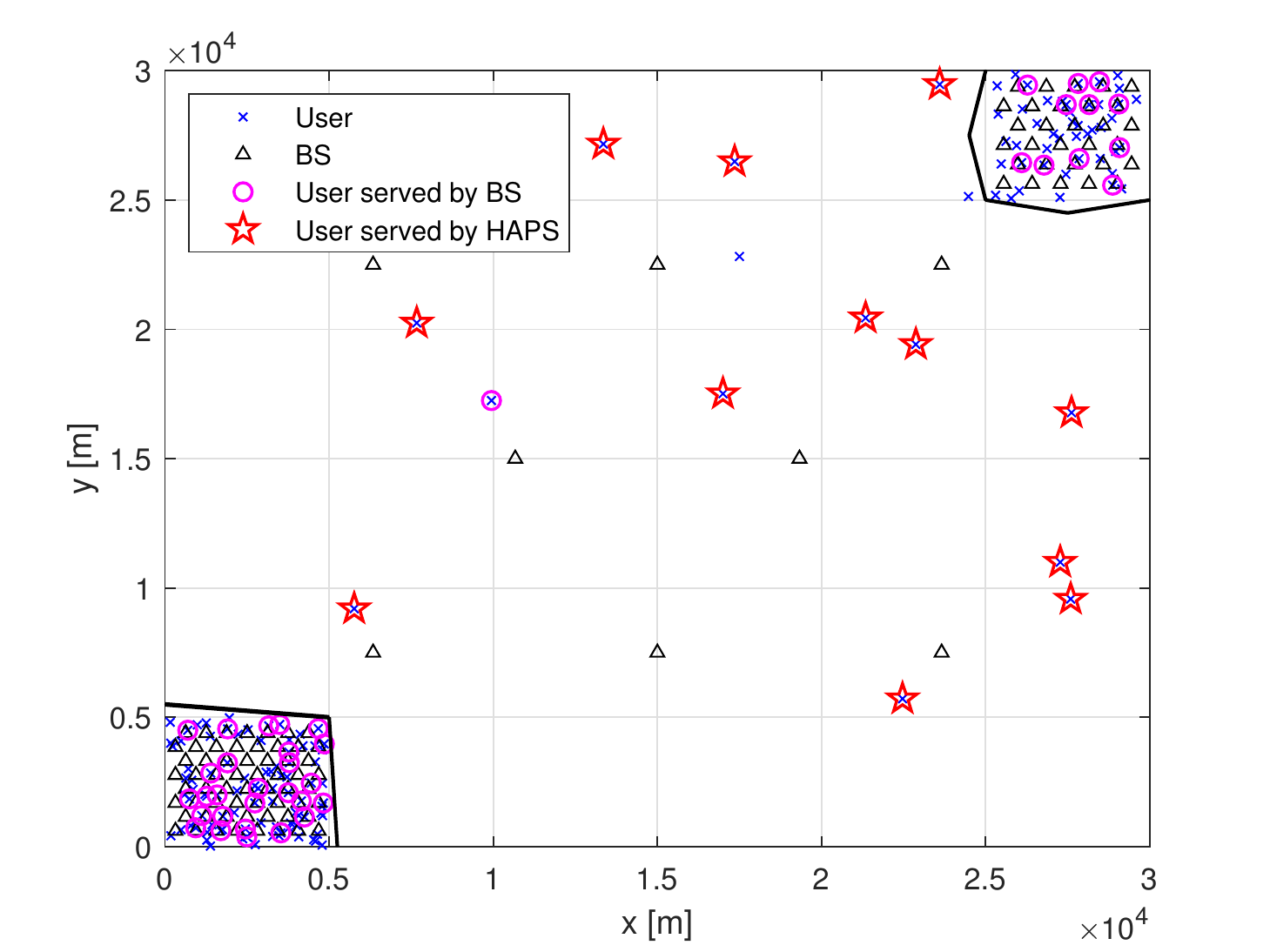}
\label{power_10dBw}
}
\subfigure[Maximum power of HAPS is 30dBw]{
\includegraphics[width=2.5in]{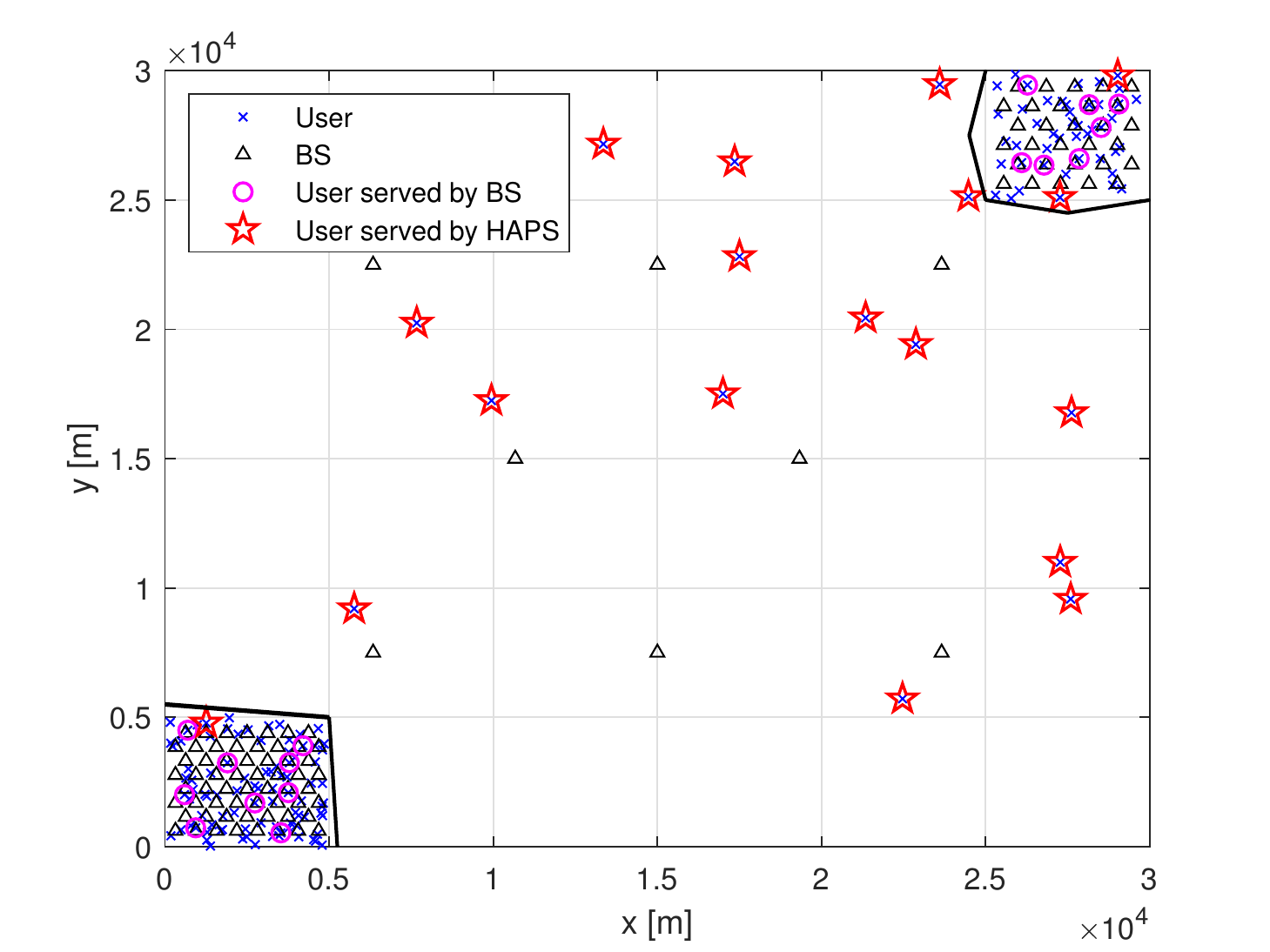}
\label{power_30dBw}
}
\caption{The behavior of user to HAPS and ground BSs association for different HAPS power levels.}
\label{how users are served}
\end{center}
\end{figure}

Fig. \ref{how users are served} shows a graphical illustration of how ground users association changes as the maximum power of the HAPS increases. The figures show that when the power of HAPS is very low, users in urban and suburban areas tend to be served by BSs in Fig. \ref{power_-40dBw}, while HAPS serves only one user in the rural area. As the power of HAPS increases to $10$dBw in Fig.~\ref{power_10dBw}, HAPS starts serving more users in the rural area, which shows how HAPS help connecting the unconnected. As the power increases further in Fig.~\ref{power_30dBw}, the HAPS starts further serving more users from within the urban area, which is an example of how HAPS ultra-connects the connected.

\begin{figure}[!t]
\centering
\includegraphics[width=3in]{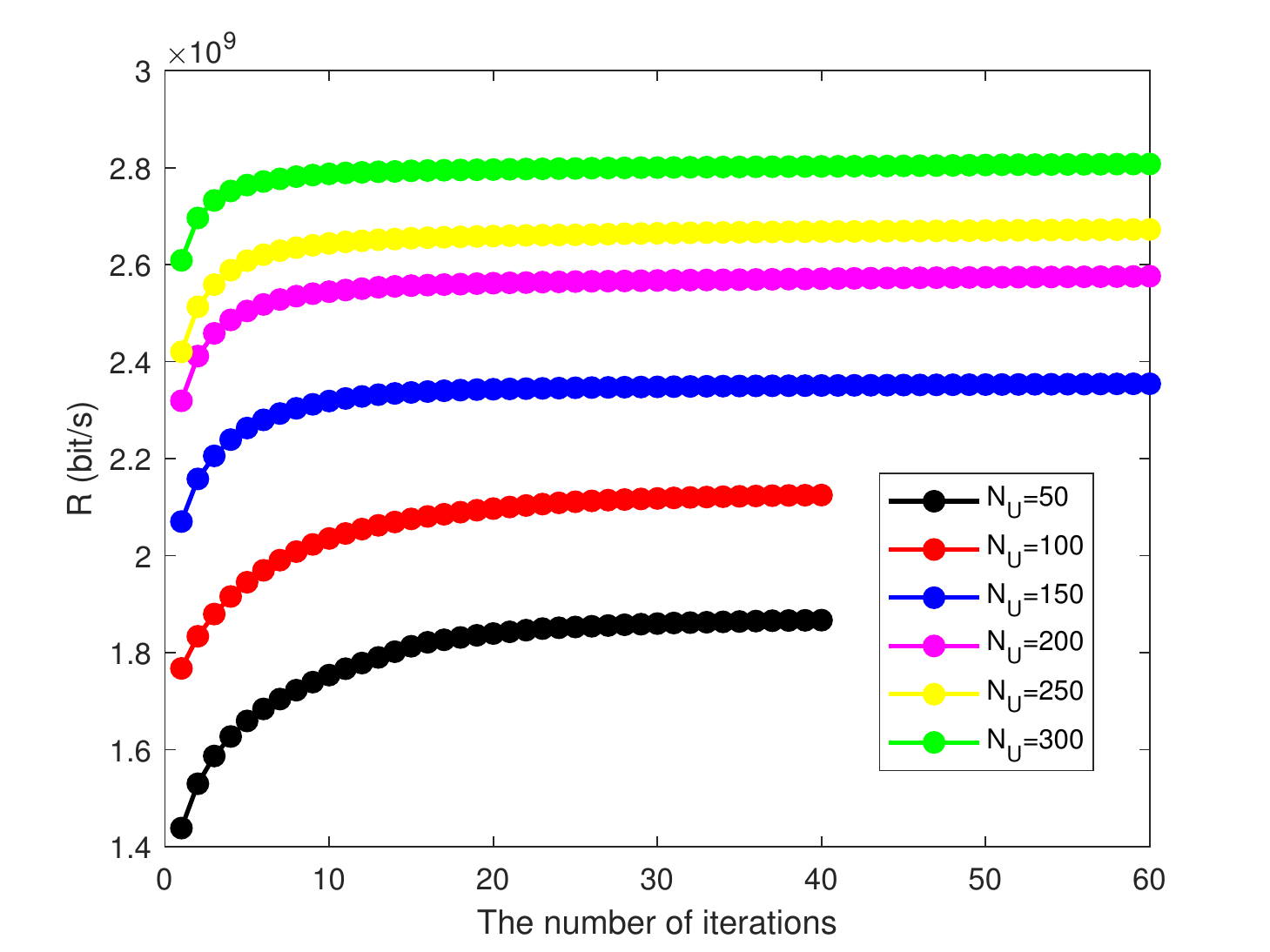}
\caption{Convergence behavior of the proposed algorithm for different numbers of users.}
\label{bigiter}
\end{figure}
Lastly, we show the convergence of our proposed algorithm through Fig. \ref{bigiter}, which plots the sum-rate behavior versus the number of iterations.
Each iteration in Fig. \ref{bigiter} consists of two inner iterations.
The figure illustrates the convergence of the proposed algorithm for various values of the number of users, which further validates the algorithmic convergence discussed earlier in the paper.
The figure particularly shows how the convergence of the overall algorithm is reasonably fast for different numbers of users, which reemphasizes the numerical prospects of the proposed algorithm in the context of our paper.

\subsection{Discussion and Recommendation}
Based on the above results, it can be seen that our proposed algorithm provides a superior sum-rate performance as compared to conventional baselines. This is particularly true at the high FSO backhaul capacity regime, and under strong HAPS capabilities (e.g., high transmit power, large number of antennas), where the HAPS offers the prospects of assisting both the unconnected and the connected; thereby improving both rural and metropolitan networks performance. A major sailing outlook of the current paper is that the considered high-speed FSO backhaul link extends from the geo-layer to the stratosphere without going through the troposphere, which makes it less vulnerable to general weather conditions. Such appealing features of the adopted FSO path, together with the general advent in HAPS design \cite{kurt2021vision}, make the paper results of particular importance in high-demand communication scenarios, e.g., concerts, sports events, etc., as well in the general context of boosting ground level future 6G networks performance.

\section{Conclusions and Future Work}
Digital inclusion is nowadays celebrated as one of the major drivers towards defining 6G communications networks sustainable architectures. Along this direction, this paper proposes an integrated satellite-HAPS-ground network consisting of one geo-satellite, one HAPS, and several terrestrial BSs that collaboratively aim at connecting the unconnected and ultra-connecting the connected. The paper focuses on maximizing the network-wide sum-rate utility, subject to {HAPS payload connectivity constraint,}
HAPS and BSs transmit power constraints, and FSO backhaul constraints, so as to jointly determine the user-association strategy of each user, and the beamforming vectors at the HAPS and BSs.
The paper tackles such a complex mixed discrete-continuous optimization problem using an iterative approach, where the user association is determined using a combination of linear integer programming and generalized assignment problems, and where the beamforming strategy is found using a WMMSE approach. The paper results illustrate the appreciable gain of our proposed algorithm, and particularly highlight the numerical prospects of augmenting the ground networks with HAPS for connecting the unconnected (through strong HAPS capabilities), and super-connecting the connected (at the high interference regime), which give promising performance projections about vertical heterogenous networks prospects. Future research directions in the field include accounting for imperfect channel information at both the HAPS and the ground base-stations, investigating the multi-HAPS scenario, and evaluating the data-driven optimization approach to solve the problem under consideration in an online fashion. Other research extensions of the current work would also include optimizing the end-to-end delay of the system by factoring in the impact of the gateway to geo-satellite feeder link latency, which promises to be a promising future research direction that falls at the intersection of communication, networking, and computing.

\appendices
\numberwithin{equation}{section}

\section{Proof of Lemma 1} \label{ap1}

\begin{proof}
Similar to the steps of \cite{shi2011iteratively}, we first rewrite the RF rate (\ref{R_RF}) as:
\begin{equation}\label{R_{n_{i}}}
  R_{ij}^{RF}=\log \mathbf{det}\left((\mathbf{e}_{ij}^{mmse})^{-1}\right), \forall i \in \mathcal{I}, \ \forall j\in\mathcal{U}_{i},
\end{equation}
where $\mathbf{e}_{ij}^{mmse}$ is the associated MSE covariance martix.

The objective of problem (\ref{EHM7_first}) (and equivalently (\ref{EHM6})) can therefore be rewritten as follows:
\begin{equation}\label{Covergence1}
  f_{1}(\mathbf{w}_{ij})=\sum_{i\in\mathcal{I}-\{0\}}\sum_{j\in\mathcal{U}_{i}}\lambda_{ij}\mathbf{det}\left ((\mathbf{e}_{ij}^{mmse})^{-1}\right)+\sum_{j\in\mathcal{U}_{0}}\lambda_{0j}\tau_{0j}.
\end{equation}
Similarly, the objective function of the equivalent problem (\ref{EHM7}) can be rewritten as:
\begin{equation}\label{Covergence2}
  f_{2}({\boldsymbol{\rho}_{ij},\mathbf{u}_{ij},\mathbf{w}_{ij}})=\sum_{i\in\mathcal{I}}\sum_{j\in\mathcal{U}_{i}}\lambda_{ij}\left(\mathbf{Tr}(\boldsymbol{\rho}_{ij}\mathbf{e}_{ij})-\log \text{det}(\boldsymbol{\rho}_{ij})\right).
\end{equation}

Since problem (\ref{EHM7}) is further differentiable, and since its constraints set is separable in the variables $\boldsymbol{\rho}_{ij},\mathbf{u}_{ij},\mathbf{w}_{ij}$, iteratively solving for one of the variables while fixing the two others, i.e., using a block coordinate descent approach, is guaranteed to converge to a stationary point \cite{shi2011iteratively}.
To finalize the proof, it suffices to show that the stationary point of (\ref{EHM7}) is the same as the stationary point of (\ref{EHM7_first}) (and equivalently (\ref{EHM6})), and that the converse is true.

Since the variables $\boldsymbol{\rho}_{ij},\mathbf{u}_{ij}$ are unconstrained, their respective first order optimality conditions yield optimal $\boldsymbol{\rho}_{ij}^{*},\mathbf{u}_{ij}^{*}$ with expressions similar to (\ref{MMSE receiver}) and (\ref{MSE_weight}), i.e.:
\begin{equation}\label{U}
\mathbf{u}_{ij}^{*}=\mathbf{u}_{ij}^{mmse},\quad  \boldsymbol{\rho}_{ij}^{*}=(\mathbf{e}_{ij}^{mmse})^{-1}, \forall i \in \mathcal{I}, \ \forall j\in\mathcal{U}_{i}.
\end{equation}
Let $\mathbf{w}_{ij, r}$ be the $r$th entry of the vector $\mathbf{w}_{ij}$, we get:
\begin{scriptsize}
\begin{eqnarray}
\label{KKT}
  \!\frac{\partial f_{2}({\boldsymbol{\rho}_{ij}^{*},\mathbf{u}_{ij}^{*},\mathbf{w}_{ij}^{*}})}{\partial \mathbf{w}_{ij, r}}\!
   &\!=\!& \!\sum_{i\in\mathcal{I}-\{0\}}\! \!\sum_{j\in\mathcal{U}_{i}}\!\lambda_{ij}\mathbf{Tr}\left((\mathbf{e}_{ij}^{mmse})^{-1}\frac{\partial\mathbf{e}_{ij}^{mmse}(\mathbf{w}_{ij}^{*})}{\partial \mathbf{w}_{ij, r}}\right) \nonumber\\
             & &
  \!+\! \sum_{j\in\mathcal{U}_{0}}\lambda_{0j}\frac{\partial \tau_{0j}}{\partial \mathbf{w}_{0j, r}}
  \label{lastequality}\\
   &\!=\!& \frac{\partial f_{1}(\mathbf{w}_{ij}^{*})}{\partial \mathbf{w}_{ij, r}},
\end{eqnarray}
\end{scriptsize}
\noindent where the second term of the equality (\ref{lastequality}) can be further developed as:
\begin{scriptsize}
\begin{equation}
\label{weight1}
 \lambda_{0j}\frac{\partial \tau_{0j}}{\partial  \mathbf{w}_{0j, r}}=
\begin{cases}
\mathbf{Tr}\left((\mathbf{e}_{0j}^{mmse})^{-1}\frac{\partial\mathbf{e}_{0j}^{mmse}(\mathbf{w}_{0j}^{*})}{\partial \mathbf{w}_{0j, r}}\right),& \ \tau_{0j}=R_{0j}^{HAPS\_RF}.\\
0,& \ \tau_{0j}=R^{FSO}.
\end{cases}
\end{equation}
\end{scriptsize}
The converse follows a reversely traversed equality path, which proves Lemma 1.
\end{proof}

\bibliography{my_bibliography}

\begin{thebibliography}{10}
\providecommand{\url}[1]{#1}
\csname url@samestyle\endcsname
\providecommand{\newblock}{\relax}
\providecommand{\bibinfo}[2]{#2}
\providecommand{\BIBentrySTDinterwordspacing}{\spaceskip=0pt\relax}
\providecommand{\BIBentryALTinterwordstretchfactor}{4}
\providecommand{\BIBentryALTinterwordspacing}{\spaceskip=\fontdimen2\font plus
\BIBentryALTinterwordstretchfactor\fontdimen3\font minus
  \fontdimen4\font\relax}
\providecommand{\BIBforeignlanguage}[2]{{%
\expandafter\ifx\csname l@#1\endcsname\relax
\typeout{** WARNING: IEEEtran.bst: No hyphenation pattern has been}%
\typeout{** loaded for the language `#1'. Using the pattern for}%
\typeout{** the default language instead.}%
\else
\language=\csname l@#1\endcsname
\fi
#2}}
\providecommand{\BIBdecl}{\relax}
\BIBdecl

\bibitem{ShashaTWC2022}
\BIBentryALTinterwordspacing
S.~Liu, H.~Dahrouj, and M.-S. Alouini, ``Joint user association and beamforming
  in integrated satellite-{HAPS}-ground networks,'' 2022. [Online]. Available:
  \url{https://arxiv.org/abs/2204.13257.}
\BIBentrySTDinterwordspacing

\bibitem{saeed2021P2P}
N.~Saeed, H.~Almorad, H.~Dahrouj, T.~Y. Al-Naffouri, J.~S. Shamma, and M.-S.
  Alouini, ``Point-to-point communication in integrated satellite-aerial 6{G}
  networks: State-of-the-art and future challenges,'' \emph{IEEE Open Journal
  of the Communications Society}, 2021.

\bibitem{alam2021high}
M.~S. Alam, G.~K. Kurt, H.~Yanikomeroglu, and P.~Zhu, ``High altitude platform
  station based super macro base station constellations,'' \emph{IEEE
  Communications Magazine}, vol.~59, no.~1, pp. 103--109, 2021.

\bibitem{qiu2019air}
J.~Qiu, D.~Grace, G.~Ding, M.~D. Zakaria, and Q.~Wu, ``Air-ground heterogeneous
  networks for 5{G} and beyond via integrating high and low altitude
  platforms,'' \emph{IEEE Wireless Communications}, vol.~26, no.~6, pp.
  140--148, 2019.

\bibitem{arum2020review}
S.~C. Arum, D.~Grace, and P.~D. Mitchell, ``A review of wireless communication
  using high-altitude platforms for extended coverage and capacity,''
  \emph{Computer Communications}, vol. 157, pp. 232--256, 2020.

\bibitem{kurt2021vision}
G.~K. Kurt, M.~G. Khoshkholgh, S.~Alfattani, A.~Ibrahim, T.~S. Darwish, M.~S.
  Alam, H.~Yanikomeroglu, and A.~Yongacoglu, ``A vision and framework for the
  high altitude platform station ({HAPS}) networks of the future,'' \emph{IEEE
  Communications Surveys \& Tutorials}, vol.~23, no.~2, pp. 729--779, 2021.

\bibitem{mohammed2011role}
A.~Mohammed, A.~Mehmood, F.-N. Pavlidou, and M.~Mohorcic, ``The role of
  high-altitude platforms ({HAPs}) in the global wireless connectivity,''
  \emph{Proceedings of the IEEE}, vol.~99, no.~11, pp. 1939--1953, 2011.

\bibitem{karapantazis2005broadband}
S.~Karapantazis and F.~Pavlidou, ``Broadband communications via high-altitude
  platforms: A survey,'' \emph{IEEE Communications Surveys \& Tutorials},
  vol.~7, no.~1, pp. 2--31, 2005.

\bibitem{fidler2010optical}
F.~Fidler, M.~Knapek, J.~Horwath, and W.~R. Leeb, ``Optical communications for
  high-altitude platforms,'' \emph{IEEE Journal of selected topics in quantum
  electronics}, vol.~16, no.~5, pp. 1058--1070, 2010.

\bibitem{henniger2010introduction}
H.~Henniger and O.~Wilfert, ``An introduction to free-space optical
  communications.'' \emph{Radioengineering}, vol.~19, no.~2, 2010.

\bibitem{liu2012achieving}
L.~Liu, R.~Zhang, and K.-C. Chua, ``Achieving global optimality for weighted
  sum-rate maximization in the {K}-user {Gaussian} interference channel with
  multiple antennas,'' \emph{IEEE Transactions on Wireless Communications},
  vol.~11, no.~5, pp. 1933--1945, 2012.

\bibitem{kim2009interference}
T.~H. Kim and S.~Choi, ``Interference mitigation via scheduling for the {MIMO}
  broadcast channel with limited feedback,'' in \emph{2009 IEEE 20th
  International Symposium on Personal, Indoor and Mobile Radio
  Communications}.\hskip 1em plus 0.5em minus 0.4em\relax IEEE, 2009, pp.
  2035--2039.

\bibitem{wang2008user}
J.~Wang, D.~J. Love, and M.~D. Zoltowski, ``User selection with zero-forcing
  beamforming achieves the asymptotically optimal sum rate,'' \emph{IEEE
  Transactions on Signal Processing}, vol.~56, no.~8, pp. 3713--3726, 2008.

\bibitem{reifert2021distributed}
R.-J. Reifert, A.~A. Ahmad, H.~Dahrouj, A.~Chaaban, A.~Sezgin, T.~Y.
  Al-Naffouri, and M.-S. Alouini, ``Distributed {Resource} {Management} in
  {Downlink} {Cache}-enabled {Multi-cloud Radio Access Networks},'' \emph{arXiv
  preprint arXiv:2104.03664}, 2021.

\bibitem{shen2018fractional}
K.~Shen and W.~Yu, ``Fractional programming for communication systems—{Part
  II: Uplink scheduling via matching},'' \emph{IEEE Transactions on Signal
  Processing}, vol.~66, no.~10, pp. 2631--2644, 2018.

\bibitem{douik2020mode}
A.~Douik, H.~Dahrouj, O.~Amin, B.~Aloquibi, T.~Y. Al-Naffouri, and M.-S.
  Alouini, ``{Mode Selection and Power Allocation in Multi-Level Cache-Enabled
  Networks},'' \emph{IEEE Communications Letters}, vol.~24, no.~8, pp.
  1789--1793, 2020.

\bibitem{douik2020tutorial}
A.~Douik, H.~Dahrouj, T.~Y. Al-Naffouri, and M.-S. Alouini, ``{A tutorial on
  clique problems in communications and signal processing},'' \emph{Proceedings
  of the IEEE}, vol. 108, no.~4, pp. 583--608, 2020.

\bibitem{Murat}
M.~Dörterler, ``A new genetic algorithm with agent-based crossover for
  generalized assignment problem,'' \emph{Information Technology And Control},
  vol.~48, pp. 389--400, 2019.

\bibitem{dahrouj2010coordinated}
H.~Dahrouj and W.~Yu, ``Coordinated beamforming for the multicell multi-antenna
  wireless system,'' \emph{IEEE transactions on wireless communications},
  vol.~9, no.~5, pp. 1748--1759, 2010.

\bibitem{dahrouj2011multicell}
H.~Dahrouj and W.~Yu, ``Multicell interference mitigation with joint
  beamforming and common message decoding,'' \emph{IEEE Transactions on
  Communications}, vol.~59, no.~8, pp. 2264--2273, 2011.

\bibitem{shi2011iteratively}
Q.~Shi, M.~Razaviyayn, Z.-Q. Luo, and C.~He, ``An iteratively weighted {MMSE}
  approach to distributed sum-utility maximization for a {MIMO} interfering
  broadcast channel,'' \emph{IEEE Transactions on Signal Processing}, vol.~59,
  no.~9, pp. 4331--4340, 2011.

\bibitem{dai2014sparse1}
B.~Dai and W.~Yu, ``Sparse beamforming and user-centric clustering for downlink
  cloud radio access network,'' \emph{IEEE Access}, vol.~2, pp. 1326--1339,
  2014.

\bibitem{shen2018fractional1}
K.~Shen and W.~Yu, ``{Fractional programming for communication systems—Part
  I: Power control and beamforming},'' \emph{IEEE Transactions on Signal
  Processing}, vol.~66, no.~10, pp. 2616--2630, 2018.

\bibitem{yu2013multicell}
W.~Yu, T.~Kwon, and C.~Shin, ``Multicell coordination via joint scheduling,
  beamforming, and power spectrum adaptation,'' \emph{IEEE Transactions on
  Wireless Communications}, vol.~12, no.~7, pp. 1--14, 2013.

\bibitem{khan2018optimizing}
A.~A. Khan, R.~Adve, and W.~Yu, ``{Optimizing multicell scheduling and
  beamforming via fractional programming and Hungarian algorithm},'' in
  \emph{2018 IEEE Globecom Workshops (GC Wkshps)}.\hskip 1em plus 0.5em minus
  0.4em\relax IEEE, 2018, pp. 1--6.

\bibitem{alzenad2019coverage}
M.~Alzenad and H.~Yanikomeroglu, ``{Coverage and rate analysis for vertical
  heterogeneous networks (VHetNets)},'' \emph{IEEE Transactions on Wireless
  Communications}, vol.~18, no.~12, pp. 5643--5657, 2019.

\bibitem{cherif2020downlink}
N.~Cherif, M.~Alzenad, H.~Yanikomeroglu, and A.~Yongacoglu, ``{Downlink
  coverage and rate analysis of an aerial user in vertical heterogeneous
  networks (VHetNets)},'' \emph{IEEE Transactions on Wireless Communications},
  vol.~20, no.~3, pp. 1501--1516, 2020.

\bibitem{jia2020sum}
H.~Jia, Y.~Wang, M.~Liu, and Y.~Chen, ``Sum-rate maximization for {UAV} aided
  wireless power transfer in space-air-ground networks,'' \emph{IEEE Access},
  vol.~8, pp. 216\,231--216\,244, 2020.

\bibitem{wang2022resource}
L.~Wang, X.~Zhao, C.~Wang, and W.~Wang, ``Resource allocation algorithm based
  on power control and dynamic transmission protocol configuration for
  {HAPS-IMT} integrated system,'' \emph{Electronics}, vol.~20, no.~1, pp.
  44--64, 2021.

\bibitem{alsharoa2020improvement}
A.~Alsharoa and M.-S. Alouini, ``Improvement of the global connectivity using
  integrated satellite-airborne-terrestrial networks with resource
  optimization,'' \emph{IEEE Transactions on Wireless Communications}, vol.~19,
  no.~8, pp. 5088--5100, 2020.

\bibitem{yahia2021haps}
O.~B. Yahia, E.~Erdogan, G.~K. Kurt, I.~Altunbas, and H.~Yanikomeroglu, ``{HAPS
  selection for hybrid RF/FSO satellite networks},'' \emph{arXiv preprint
  arXiv:2107.12638}, 2021.

\bibitem{alzenad2018fso}
M.~Alzenad, M.~Z. Shakir, H.~Yanikomeroglu, and M.-S. Alouini, ``{FSO}-based
  vertical backhaul/fronthaul framework for 5{G}+ wireless networks,''
  \emph{IEEE Communications Magazine}, vol.~56, no.~1, pp. 218--224, 2018.

\bibitem{ganian2019solving}
R.~Ganian and S.~Ordyniak, ``Solving integer linear programs by exploiting
  variable-constraint interactions: A survey,'' \emph{Algorithms}, vol.~12,
  no.~12, pp. 248--261, 2019.

\bibitem{ross1975branch}
G.~T. Ross and R.~M. Soland, ``{A branch and bound algorithm for the
  generalized assignment problem},'' \emph{Mathematical programming}, vol.~8,
  no.~1, pp. 91--103, 1975.

\bibitem{boyd2004convex}
S.~Boyd and L.~Vandenberghe, \emph{Convex optimization}.\hskip 1em plus 0.5em
  minus 0.4em\relax Cambridge university press, 2004.

\bibitem{dahrouj2015distributed}
H.~Dahrouj, T.~Y. Al-Naffouri, and M.-S. Alouini, ``Distributed cloud
  association in downlink multicloud radio access networks,'' in \emph{2015
  49th Annual Conference on Information Sciences and Systems (CISS)}.\hskip 1em
  plus 0.5em minus 0.4em\relax IEEE, 2015, pp. 1--3.

\end{thebibliography}
\bibliographystyle{IEEEtran}
\end{document}